\documentclass[11pt]{article}
\usepackage{graphicx}
\usepackage{subcaption} 
\usepackage{geometry}
\usepackage[T1]{fontenc}
\usepackage[utf8]{inputenc} 

\usepackage{url}

\usepackage[authoryear,round]{natbib}

\usepackage[english]{babel}

\usepackage{thmtools,thm-restate} 
\usepackage{bbm}
\usepackage{amsthm,amsmath,amssymb}
\usepackage{graphicx}
\usepackage{enumerate}
\usepackage[dvipsnames]{xcolor}
\usepackage{url}

\usepackage[ruled]{algorithm2e} 

\SetAlFnt{\small}
\SetAlCapFnt{\small}
\SetAlCapNameFnt{\small}
\SetAlCapHSkip{0pt}
\IncMargin{-1.6\parindent}

\usepackage[
  colorlinks=true,
  citecolor=blue,
  linkcolor=blue,
  urlcolor=blue,
  plainpages=false,
  hypertexnames=false
]{hyperref}

\theoremstyle{plain}
\newtheorem{theorem}{Theorem}
\newtheorem{lemma}[theorem]{Lemma}

\newtheorem{proposition}[theorem]{Proposition}

\theoremstyle{definition}
\newtheorem{definition}[theorem]{Definition}

\theoremstyle{remark}

\usepackage{tikz}
\usetikzlibrary{arrows,automata,positioning}
\usepackage[utf8]{inputenc}
\usepackage[T1]{fontenc}
\usepackage{pgfplots}
\pgfplotsset{compat=1.13}

\usetikzlibrary{math}
\tikzmath{
\f = 0.5; \g = (\f*(1-\f))^0.5;
\aval = 0.15; \cval = 1-\aval; \bval = (\aval*(1-\aval))^0.5; 
\q = 0.6; \r = (\q*(1-\q))^0.5;
\uval = 0.3; \vval = 1-\uval; \wval = (\uval*(1-\uval))^0.5; 
\s = 0.92; \t = (\s*(1-\s))^0.5;
\scale = 2;
\stretch=0.3;
\stretchh=1;
}

\usepackage{comment}
\usepackage{relsize}

\newcommand{\N}{\mathbb{N}}
\newcommand{\R}{\mathbb{R}}

\newcommand{\Prob}{\mathbb{P}}
\newcommand{\Pmud}{\mathcal{P}(\mu,d)}
\newcommand{\Ptwomud}{\mathcal{P}_2(\mu,d)}

\newcommand{\REV}{\text{REV}}

\usepackage{dsfont}
\newcommand{\1}{\mathds{1}}
\newcommand{\bP}{\mathbb{P}}
\newcommand{\cP}{\mathcal{P}}

\usepackage{ifthen}

\definecolor{forestgreen}{rgb}{0.13, 0.55, 0.13}

\numberwithin{theorem}{section}
\numberwithin{equation}{section}
\begin{document}
\title{
Robust Optimality of Bundling Goods \\ Beyond Finite Variance
}
\author{
Tim S.\,G.\ van Eck\textsuperscript{1},
Pieter Kleer\textsuperscript{1},
and Johan S.\,H.\ van Leeuwaarden\textsuperscript{1} \\
\\[-2pt]
\textsuperscript{1} Tilburg University, Tilburg, The Netherlands \\
\texttt{$\{$t.s.g.vaneck,p.s.kleer,j.s.h.vanleeuwaarden$\}$@tilburguniversity.edu}
}

	\maketitle

\thispagestyle{empty}

\begin{abstract}
    Buyer valuations for a collection of goods may exhibit heavy tails, while sellers often know only limited information about their distribution. In such settings, full distributional assumptions or reliance on second moments may be inappropriate. We assume independence across goods and knowledge of the mean and mean absolute deviation (MAD) of valuations, which keeps the model well defined even under heavy-tailed uncertainty. Within this distributionally robust framework, the seller chooses a selling mechanism to maximize worst-case expected revenue. We show that, in the many-goods limit, no selling mechanism can guarantee extraction of the full mean revenue per good. This contrasts sharply with the finite-variance case, where deterministically selling all goods as a single bundle achieves full mean extraction asymptotically. Despite this shortfall, we show that deterministic bundling remains optimal. Moreover, we show that this bundling price is also optimal when restricting attention to deterministic bundling mechanisms under the competitive ratio and absolute regret objectives, leading to a simple and universal pricing rule.
\end{abstract}

\newpage

\setcounter{page}{1}
\section{Introduction}
In many markets for digital goods firms increasingly rely on bundling rather than selling goods exclusively on an à-la-carte basis. Platforms that historically emphasized individual purchases now often offer access to large collections of goods through bundled pricing schemes, such as subscription-based access to movies, music, or video games. As a result, understanding bundling as a selling mechanism has become an important topic in economics and mechanism design, and a substantial literature has studied its revenue properties and conditions for optimality.
A common feature of these markets is the presence of large catalogs of goods combined with negligible marginal costs of supply. While bundling is already economically interesting with only a few goods, the size of modern catalogs naturally raises the question of how revenue guarantees scale with the number of goods. We therefore study a single-buyer, multi-good revenue maximization problem in the many-goods limit, which allows us to characterize how revenue guarantees scale with the number of goods.

Within this many-goods limiting setting, \cite{bakos1999bundling} show that when the distributions of the goods are i.i.d., fully known, and do not depend on the number of goods, deterministic bundling becomes optimal. This also means that randomized mechanisms cannot improve the expected revenue in this environment, which they typically can in other settings. Using the law of large numbers together with Chebyshev's inequality, they establish that the seller's expected revenue per good converges to the mean, which serves as an upper bound for any truthful selling mechanism. Truthfulness refers to the fact that the buyer should have no incentive to lie about the valuation they have for the good. By suitably truncating the distribution, they further extend this result to heavy-tailed distributions with infinite variance.

The result of \cite{bakos1999bundling} shows the strength of deterministic bundling when full information on the valuation distributions is available. However, full information is often \textit{not} available or unreliable and therefore deemed unrealistic in practice. Adding to this is what is commonly known as Wilson's doctrine \citep{Wilson1987, Milgrom2004}: We should aim at designing mechanisms that only require minimal assumptions about the value distribution of the buyer. This leads to the notion of \textit{distributionally robust revenue maximization}, in which the seller only has \textit{partial knowledge} about the value distributions of the buyer, often in the form of summary statistics such as the mean and the variance (see, e.g., \cite{scarf1958min}). We therefore adopt a distributionally robust approach, assuming independent valuations and partial knowledge of the mean and the mean absolute deviation (MAD).

By focusing on the MAD rather than the variance, we depart from much of the existing literature, which typically centers on partial information that ensures finite variance. While this is analytically convenient, this excludes distributions with heavy tails that have infinite variance. There is ample empirical evidence for heavy-tailed valuations and demand, also spurred by ``the long-tail phenomenon'' seen in online retail, where niche products often determine a large share of the total sales \citep{das2021heavy}. Examples of reported heavy tails include demand for books \citep{gaffeo2008demand}, movies \citep{bimpikis2016inventory}, spare parts \citep{natarajan2018asymmetry}, and many more \citep{clauset2009power}, where the tail behavior frequently yields a finite mean but infinite variance. Such heavy-tailed behavior can substantially influence revenue outcomes when only partial distributional information is known.

Our first contribution shows that, in the many-goods limit, \textit{all} truthful mechanisms achieve expected revenue per good that is strictly below the mean. This stands in stark contrast to the mean-variance setting \citep{giannakopoulos2023ratio}, where a concentration argument (Chebyshev's inequality) shows that deterministic bundling yields the full mean, thereby capturing the full market potential similar to the full-information model. Such arguments break down entirely when the variance may be infinite. We then show, as our second contribution, that deterministic bundling remains optimal within this setting despite not capturing the full market potential. Hence, deterministic bundling remains a powerful mechanism when distributional ambiguity and heavy-tails interact.

We also consider deterministic bundling for the competitive ratio and absolute regret objectives, which compare the expected revenue with the optimal expected revenue under full information. The competitive ratio evaluates this comparison in relative terms, whereas absolute regret evaluates it in absolute terms. As our third contribution, we show that, among all deterministic bundling prices, the same price is optimal for expected revenue, competitive ratio, and absolute regret in the many-goods limit. An important motivation for us to consider these alternative objectives is a recent study by \cite{anunrojwong2024best}, which investigates robust decision making across all of the aforementioned criteria for the single-good setting. Their central insight is that mechanisms optimized for one objective often perform poorly under another, revealing a form of objective overfitting. They propose a unified approach that directly optimizes for performance across all objectives simultaneously. Our results contribute to this line of inquiry by offering a first step toward understanding objective alignment in a canonical high-dimensional setting.

From a technical perspective, establishing our results requires addressing optimization problems that, because the goods are independent, are much harder than its correlated counterparts. Under correlated valuations, one can resort to well-known primal-dual semi-infinite linear programming techniques or computing a saddle point of a zero-sum game. Independence, from an optimization point of view, results in non-linear (and non-convex) problems. A fitting analogy is comparing the difference in computational complexity between computing a mixed Nash equilibrium (difficult) and a correlated equilibrium in general finite games (tractable via a linear program). Independence thus introduces a substantive layer of complexity into the analysis.

\subsection{Further related work}
In this section, we discuss related works, starting with the classical Bayesian setting. We then continue with robust mechanisms for selling multiple goods in other settings, as well as the literature that addresses the case when there are multiple buyers. Finally, we discuss related works for robustly selling a single good. These works demonstrate the growing interest in obtaining distributionally robust solutions to classical mechanism design problems, to which our work contributes.\vspace{0.6em}

\label{sec:related}
\noindent \textit{Mechanisms for selling multiple goods in the full-information setting.}
In the case of selling $m=1$ good with a given valuation distribution, the optimal mechanism is a posted "take-it-or-leave-it" price \citep{myerson1981optimal,riley1981optimal,riley1983optimal} which sells the good if and only if the buyer's (random) valuation for the good exceeds this price. When selling $m>1$ goods, the optimal mechanism may exhibit a much more complex structure \citep{daskalakis2017strong,giannakopoulos2018duality}. Natural choices like bundling all goods together or selling them separately are in general not optimal for any distribution. Furthermore, the optimal mechanism need not be deterministic, as opposed to the single-good case. These observations are nicely summarized by \cite[Examples 1-4]{hart2017approximate}. In the many-goods setting with $m\xrightarrow{} \infty$, \cite{bakos1999bundling} show that bundling is optimal as long as the valuation distribution does not depend on $m$. When the distribution is allowed to depend on $m$, bundling is no longer guaranteed to be optimal. Instead, \cite{hart2017approximate} show that bundling achieves a $c/\log(m)$-fraction of the revenue of the optimal mechanism for some constant $c>0$. Later, \cite{li2013revenue} improved this result to show that bundling guarantees a constant fraction of the revenue of the (unknown) optimal mechanism. For finite $m$, \cite{daskalakis2017strong} provide a primal-dual framework that allows them to characterize the optimal mechanism for selling a fixed number of $m \geq 2$ goods, thereby also generalizing earlier works such as \cite{manelli2006bundling,cai2012algorithmic}; see also \cite{cai2019duality} for a more general framework. \cite{daskalakis2017strong} also provide a characterization for when bundling is optimal. As a concrete application of this characterization, they consider i.i.d. goods with uniform distribution on $[c,c+1]$ for some scalar $c$. They show that if $m$ is kept fixed, then for $c$ large enough the bundling mechanism is optimal (this was shown by \cite{pavlov2011optimal} for $m = 2$). However, for $c$ fixed, it is shown that for $m$ large enough, the bundling mechanism is \emph{not} optimal. 
In general, there is not a simple economic interpretation of the characterization in \cite{daskalakis2017strong}.
There have also been other works concerned with deriving conditions on the known independently distributed goods for when bundling is the optimal mechanism, such as \cite{haghpanah2021pure,pavlov2011optimal,daskalakis2017strong,menicucci2015optimality,manelli2006bundling,ibragimov2010optimal}. See \cite[Footnote 1]{che2021multiplegoods} for an overview. 

\vspace{0.6em}
\noindent \textit{Robust mechanisms for selling multiple goods.} \cite{carroll2017robustness} introduces a robust perspective for selling multiple goods. It is shown that when correlations between the goods are allowed, the optimal robust mechanism is to sell every good separately. See also \cite{carroll2019robustness} for an interesting survey. \cite{che2021multiplegoods} generalize Carroll's work by introducing a partition structure. They then assume to know the mean of every good, and some dispersion information about every bundle (not every individual good within a bundle) in the partition, but not the possible correlations between the goods in different parts of the partition. \cite{deb2021multi} consider a robust setting related to ours, but instead of moment information assume to have only support information for the unknown distributions, i.e., they only assume that the distributions have positive probability mass on a given interval. They prove that randomized bundling is the optimal mechanism for selling $m$ goods under an exchangeability assumption (that also applies to independent distributions); see also \cite{che2021multiplegoods} for related results in this direction. \cite{giannakopoulos2023ratio} consider the robust ratio objective and provide various tight, up to constants, results for selling multiple goods under mean-variance ambiguity for both correlated and independent value distributions. See also \cite{chen2022distribution} for results in this direction. There are also works that study (versions of) revenue maximization with respect to distance-based ambiguity sets, such as \cite{brustle2020} who use the Prokhorov and L\'evy distances in learning-based settings.
\medskip

\noindent \textit{Robust mechanisms for multiple buyers.} There are also various papers that study the setting with multiple buyers where there are unknown distributions coming from (moment-based) ambiguity sets. \cite{anunrojwong2022robustness} study the problem of selling one good to $n$ potential buyers, where the unknown (joint) distribution is only assumed to have bounded support, meaning an upper bound on the valuations is known. \cite{bachrach2022distributional} consider a similar model for selling a single good to $n$ i.i.d. (potential) buyers, but in addition assume to know the mean of the common distribution. In both these works (in \cite{bachrach2022distributional} only for $n = 2$), a robust version of the second price auction with (randomized) reserve price arises as the optimal robust solution, which is proved by using a saddle-point argument for the corresponding robust zero-sum game between the seller and Nature. For other works in the multi-buyer setting, see, e.g., \cite{koccyiugit2020distributionally,allouah2020prior,carroll2019robustness}.
\medskip

\noindent \textit{Robust mechanisms for single good.} 
Maximin analysis for mean-variance ambiguity was pioneered in \cite{azar2012optimal}, generalized to higher moments in 
\cite{carrasco2018optimal}; see also \cite{chen2022distribution,suzdaltsev2018distributionally}. Other forms of knowledge were also studied, such as percentile and quantile-based information \citep{eren2010monopoly,ge2025optimal}, mean absolute deviation \citep{roos2019chebyshev, elmachtoub2021value,chen2023screening},
or knowing that the valuation distribution is within the proximity of a given reference distribution \citep{bergemann2011robust,chen2023screening}. Instead of maximin expected revenue, alternative objectives studied in the literature include absolute regret \citep{bergemann2008pricing,bergemann2011robust,elmachtoub2021value, chen2023screening} and the competitive ratio \citep{eren2010monopoly,giannakopoulos2023ratio,wang2025power}. We also refer to the work of \cite{allouah2023optimal} who assume to have prior knowledge about the probability of sale of a good.
In some works in the literature, randomized algorithms are also considered and shown to be optimal in certain cases. For selling one good in the robust setting, \cite{chen2023screening} show, based on a functional version of Von Neumann's minimax theorem \citep{borwein1986fan}, that if randomized mechanisms are allowed in the maximin problem, the values of the maximin and minimax problems are equal to each other. We are not aware of a similar result for the problem of selling multiple goods, in particular not with the additional assumption of independent goods. This relation fails to hold if only deterministic mechanisms are allowed.

\subsection{Outline}
Section \ref{sec:pre} introduces the model and gives the necessary background on mechanism design and distributionally robust optimization. Section \ref{sec:main} presents our main results, stated in three theorems. Their proofs are presented in Sections \ref{sec:minimax}, \ref{sec:maximin}, and \ref{sec:alignment}, respectively. We then reflect on the single-good setting in Section \ref{sec:single_item}. Finally, Section \ref{sec:conclusion} offers concluding remarks and discusses possible avenues for future research.

\section{Model and preliminaries}
\label{sec:pre}
We will start by formally defining deterministic truthful mechanisms and the bundling mechanism. We then introduce the distributionally robust framework using probability theory.

\subsection{Mechanism design}
\label{sec:md}
A direct revelation mechanism $A$ in the setting of a single buyer that bids on $m$ goods is defined by a pair $(z,\pi)$. Here $z: \R_{\geq 0}^m \rightarrow [0,1]^m$ is the (randomized) \emph{allocation rule}, and $\pi : \R_{\geq 0}^m \rightarrow \R_{\geq 0}$ the \emph{payment rule}. For a given valuation vector $v = (v_1,\dots,v_m)$ with $v_i$ the value the buyer has for good $i$, they receive the goods $i$ with probability $z_i(v)$, and get charged $\pi(v)$ in total for all the goods they receive. The utility of the buyer under mechanism $A$ with valuation vector $v$ is then given by $u(A,v)  =  \langle z(v),v  \rangle - \pi(v)$ with $\langle x,y \rangle = \sum_{i=1}^m x_iy_i$ the inner product for two vectors $x,y \in \R^m$.

A mechanism $A = (z,\pi)$ is \emph{truthful} if the following conditions are satisfied:
$$
\begin{array}{ll}
1. \  \langle z(v),v  \rangle - \pi(v) \geq  \langle z(w),v\rangle - \pi(w) & \text{ for all } v,w \in \R^m_{\geq 0}, \\
2. \  \langle z(v),v \rangle - \pi(v) \geq 0 & \text{ for all } v \in \R^m_{\geq 0}.
\end{array}
$$
The first condition ensures that the buyer has no incentive to misreport the true values $v$ that they have for the goods; that is, bidding truthfully is a \textit{dominant strategy} that maximizes their utility. We refer to this property as \textit{incentive compatibility}. The second condition guarantees that if a buyer truthfully reports their values, then their utility (or surplus) is non-negative, i.e., they don't lose anything by participating in the selling mechanism. This property is called \textit{individual rationality}. We restrict to truthful mechanisms throughout the paper and denote the set of all such mechanisms for $m$ goods by $\mathcal{A}_m$. Among these mechanisms, deterministic bundling is of particular importance for our analysis.
\medskip

\noindent \textit{Deterministic bundling.} A deterministic \textit{bundling mechanism} $D_{\textup{bund}}$ sells either all goods together at a \emph{bundling price} if the sum of the valuations meets or exceeds this price; otherwise, it sells none of the goods. To be precise, it sets a price $p_m =: \pi(v)$ and defines 
$$
z(v) = \left\{ \begin{array}{ll}
(1,1,\dots,1) & \text{ if } \sum_i v_i \geq p_m \\
(0,0,\dots,0) & \text{ if } \sum_i v_i < p_m
\end{array} \right.. 
$$

\subsection{Distributionally robust framework}
\label{sec:robust_design}
In this work, we assume that the valuations of the buyer are random, denoted by the random vector $X = (X_1,\dots,X_m)$ over which we have a joint probability distribution $\Prob^m = \Pi_{i=1}^m \mathbb{P}_i$. We assume that the $X_i$ are \emph{independently distributed} and indicate their marginal distribution by $\mathbb{P}_i$, i.e., $\Prob^m$ is the product distribution $\mathbb{P}_1 \times \mathbb{P}_2 \times \dots \times \mathbb{P}_m$. The (expected) revenue of a mechanism $A \in \mathcal{A}_m$ used for selling $m$ independent goods $(X_1,\dots,X_m) \sim \Prob^m$ is given by
\begin{align}
\mathrm{REV}(A,\Prob^m) = \mathbb{E}_{\Prob^m}[\pi(X_1,\dots,X_m)].
\label{eq:revenue}
\end{align}
For the special case of the deterministic bundling mechanism that sells all goods if the sum of their values exceeds $p_m$, the revenue is denoted by
\begin{align}
\text{BUND}(p_m,\Prob^m) = \mathrm{REV}(D_{\textup{bund}},\Prob^m) = p_m \cdot \Prob^m\left(\sum_{i=1}^{m} X_i \geq p_m\right).
\label{eq:rev_bundling}
\end{align}
The robust revenue maximization problem we want to solve can be seen as a two-player game, in which the seller (first player) needs to choose a truthful mechanism $A \in \mathcal{A}_m$ with the aim of maximizing expected revenue of selling $m$ goods, after which Nature (the second player) chooses a distribution $\Prob^m = \Pi_{i=1}^m \mathbb{P}_i \in \cP(\mu,d)^m$ with the aim of minimizing (normalized) expected revenue; see \eqref{eq:maxmin_intro} below. This joint distribution $\Prob^m = \Pi_{i=1}^m\mathbb{P}_i$ is the product of the marginal distributions of all $m$ goods where each $\mathbb{P}_i$ is from the ambiguity set $\cP(\mu,d)$ that assumes to know the mean $\mu$ and mean absolute deviation (MAD) $d$ of the distribution (referred to as \textit{mean-MAD ambiguity}), i.e., 
\begin{align}
\mathcal{P}(\mu,d) = \{\mathbb{P}: \mathbb{E}_{\mathbb{P}}[X] = \mu, \, \mathbb{E}_{\mathbb{P}}[|X - \mu|] = d, \text{ and } X \in [0,\infty)\}.
\label{eq:ambi_intro}
\end{align}
The ambiguity set $\cP(\mu,d)$ is non-empty if and only if $0 \leq d < 2\mu$. Since our results are monotone in the parameters $\mu$ and $d$, they continue to hold when $\mu$ is replaced by a lower bound and $d$ by an upper bound. Furthermore, although each marginal distribution is contained in the same ambiguity set, we do not assume them to be identical, so the product joint distribution $\bP^m = \Pi_{i=1}^m \bP_i$ may consist of different (marginal) $\bP_i \in \cP(\mu,d)$. We remark, however, that all results in this paper remain valid under the additional assumption that the marginals are i.i.d.

We often consider the subclass of two-point distributions with mean $\mu$ and mean absolute deviation (MAD) $d$. This class can be fully parameterized by a single parameter $\alpha \in [d/2\mu,1)$ on the left support point. We denote this family (of two-point distributions) as $\mathcal{P}_2(\mu,d) \subset \mathcal{P}(\mu,d)$. For $\alpha \in [d/(2\mu),1)$, the distribution of a random variable $X(\alpha) \sim \mathbb{P}_\alpha \in \mathcal{P}_2(\mu,d)$ is given by
\begin{align}
X(\alpha) = \left\{ \begin{array}{ll}
x(\alpha) = \displaystyle \mu - \frac{d}{2\alpha} & \text{ w.p. } \alpha \\
y(\alpha) = \mu + \displaystyle \frac{d}{2(1-\alpha)} & \text{ w.p. } 1 - \alpha 
\end{array}\right..
\label{eq:two_point_mean_mad}
\end{align}
When $X_1,\dots,X_m$ are i.i.d. according to $\mathbb{P}_{\alpha}$, then $Y = \sum_{i=1}^m X_i \sim \Prob^m_{\alpha}$ with
\begin{align}
\Prob^m_{\alpha}(Y = (m - k) \cdot x(\alpha) + k \cdot y(\alpha)) =  \alpha^{m-k}(1-\alpha)^k\binom{m}{k} \ \ \text{ for } \ \  k = 0,\dots,m. 
\label{eq:sum_Xi}
\end{align}

Because in this work we want to let $m\xrightarrow{}\infty$, we will in fact maximize the revenue scaled by $m$. That is, we want to solve the maximin problem
\begin{align}
\sup_{A \in \mathcal{A}_m} \inf_{\Prob^m \in \cP(\mu,d)^m} \, \frac{\REV(A,\Prob^m)}{m}.
\label{eq:maxmin_intro}
\end{align} 
Closely related to \eqref{eq:maxmin_intro} is the minimax problem
\begin{align}
\inf_{\Prob^m \in \cP(\mu,d)^m} \sup_{A \in \mathcal{A}_m} \, \frac{\REV(A,\Prob^m)}{m}
\label{eq:minmax_intro}
\end{align} 
in which the roles of the players are reversed: First nature gets to choose a distribution, after which the seller has to choose an optimal selling mechanism for that distribution. Roughly speaking, the difference between the two problems lies in the fact that in the maximin problem, the seller needs to hedge against distributional uncertainty, whereas in the minimax problem the seller has to hedge against a (known) worst-case distribution.
While for $m = 1$, it is known that the value of \eqref{eq:maxmin_intro} and \eqref{eq:minmax_intro} are equal if the ambiguity set is convex \citep{chen2023screening}, this is generally not known for $m > 1$. Furthermore, equality is typically not true if one would restrict to deterministic mechanisms (we return to these observations later).

\section{Main results}\label{sec:main}
We first establish an upper bound of $\mu - d/2$ on achievable (normalized) revenue in the maximin problem \eqref{eq:maxmin_intro} as $m \rightarrow \infty$, showing that no truthful mechanism can extract the full market potential $\mu$ asymptotically. We do this by showing the stronger result that the minimax problem in \eqref{eq:minmax_intro}, that upper bounds the maximin problem, is upper bounded by $\mu - d/2$.

We then analyze the optimal revenue achievable by robust deterministic bundling and show that it also equals $\mu - d/2$ as $m \rightarrow \infty$. Because this value is a lower bound on the maximin problem, this yields that the maximin and minimax problem are asymptotically equal to each other with a common value of $\mu - d/2$. We conclude by examining the power of deterministic bundling by showing that the optimal bundling price for the expected revenue objective remains optimal among all deterministic bundling mechanisms for the competitive ratio and absolute regret objective.

\subsection{Selling below the mean}\label{sec:below_mean}
As our first contribution, Theorem \ref{thm:below_mu} shows that, in the many-goods limit, no truthful mechanism can achieve the mean as expected revenue per good. This creates a surprising contrast with the mean-variance ambiguity set, where a straightforward concentration argument for bundling yields the common value of $\mu$.

\begin{theorem}[Minimax Revenue Theorem]
Let $0 \leq d < 2\mu$. Then
\begin{align}\label{eq:bound}
\lim_{m \rightarrow \infty}   \displaystyle \inf_{\Prob^m \in \mathcal{P}(\mu,d)^m}\sup_{A\in \mathcal{A}_m} \, \frac{\mathrm{REV}(A,\Prob^m)}{m} \leq \mu-\frac{d}{2}.
\end{align}
\label{thm:below_mu}
\end{theorem}

We will next provide a sketch of the main ideas underlying the proof. To analyse the minimax problem in \eqref{eq:bound}, we plug in chosen i.i.d.~distributions and argue that no (randomized) mechanism can obtain a revenue of more than $\mu - d/2$ as $m$ grows large. To keep the analysis tractable, a first idea is to plug in a two-point distribution $\Prob_2$ supported on $\{x,y\}$ with $\Prob_2(X = x) = \alpha$ and $\Prob_2(X = y) = 1-\alpha$. Note that this implies that any realized valuation vector $v$ of the joint distribution over the $m$ goods satisfies $v \in \{x,y\}^m$.  Because of the mean and MAD constraint that $\Prob_2$ must satisfy, we can fully parameterize two-point distributions in $\mathcal{P}(\mu,d)$ by $\alpha$, i.e., $x = x(\alpha)$ and $y = y(\alpha)$.

Intuitively, the worst-case distribution for the minimax problem corresponds to the situation where $\alpha \rightarrow 1$, which implies that $x(\alpha) \rightarrow \mu - d/2$ (from the left) and $y(\alpha) \rightarrow \infty$. To achieve this, the parameter $\alpha$ is chosen to depend on $m$, i.e., $\alpha = \alpha(m)$, in such a way that $\alpha(m) \rightarrow 1$ as $m \rightarrow \infty$.

By sending $\alpha(m) \rightarrow 1$ at a sufficiently fast speed, the combined contribution of valuation vectors with two or more $y$-valued goods is negligible in the revenue analysis, and so the focus will be on analyzing the contribution of valuation vectors with at most one $y$-valued good. Intuitively, one can argue that the better of setting a deterministic bundling prices just below $m \cdot x(\alpha)$, leading to a revenue of $\mu - d/2$, or a price just below $(m-1) \cdot x(\alpha) + y(\alpha)$, leading to a revenue of $d/2$, is an upper bound on what any mechanism can achieve. In other words, this yields an upper bound of $\max\{\mu-d/2,d/2\}$. For $d < \mu$, so that $\mu - d/2 < d/2$, this gives the upper bound of $\mu - d/2$ on the achievable revenue of any mechanism.

To establish the same upper bound for $d > \mu$, it turns out that using a two-point distribution does not suffice. We illustrate this in Appendix \ref{app:two-point_minimax}. Instead, we use a more complex distribution that spreads out the probability mass of the point $y(\alpha)$ over $N$ different points $\{x_1(\alpha),\dots,x_N(\alpha)\}$, with $N$ independent of $m$, in such a way that  all points diverge to infinity, but at different speeds as $m \rightarrow \infty$.  By doing this carefully the revenue is now, again intuitively, upper bounded by either bundling with a price just below $m \cdot x(\alpha)$ or bundling just below $(m-1)x(\alpha) + x_i(\alpha)$ for some $i = 1,\dots,N$. Now, however, the revenue in the second case(s) will be at most $d/(2N)$. By choosing $N$ sufficiently large, this means that the first case becomes dominant, thereby establishing the bound of $\mu - d/2$ for $d > \mu$. The full proof is given in Section \ref{sec:minimax}. \medskip

To complement Theorem \ref{thm:below_mu} we show that deterministic bundling achieves a revenue of $\mu - d/2$ for the maximin problem, thereby establishing that deterministic bundling is optimal in the many-goods limit, and also that the asymptotic maximin and minimax problems have a common value.

\begin{theorem}[Maximin Bundling Theorem]
Let  $0 \leq d < 2\mu$. Then
\begin{align}\label{eq:mu-d/2}
\lim_{m \rightarrow \infty}   \displaystyle \sup_{p_m \geq 0}\, \inf_{\Prob^m \in \mathcal{P}(\mu,d)^m} \, \frac{\mathrm{BUND}(p_m,\Prob^m)}{m} \displaystyle \geq  \mu-\frac{d}{2},
\end{align}
which is asymptotically attained by price
\begin{align}\label{def:price}
    p_m^* = (1-\epsilon_m)^2 m\left(\mu - \frac{d}{2(1-\epsilon_m)}\right),
\end{align}
where $\epsilon_m = \frac{1}{2}(1-\frac{d}{2\mu})m^{-1/5}$.
\label{thm:mu_d_2mu}
\end{theorem}

The proof of Theorem \ref{thm:mu_d_2mu} consists of two main steps. First, we derive a one-sided concentration bound tailored to distributions with finite mean and MAD. Roughly speaking, it states that the sum of $m$ i.i.d. random variables with fixed mean and MAD will be greater than or equal to $\approx m(\mu-d/2)$ with probability close to one as $m \xrightarrow{}\infty$. In the second step, we use this concentration bound to prove the result. The full proof is given in Section \ref{sec:maximin}.

Because bundling attains the upper bound in \eqref{eq:bound}, the deterministic maximin and minimax problems coincide in the many-goods limiting setting:
$$\lim_{m \rightarrow \infty}   \displaystyle \sup_{p_m \geq 0}\, \inf_{\Prob^m \in \mathcal{P}(\mu,d)^m} \, \frac{\mathrm{REV}(p_m,\Prob^m)}{m} \displaystyle =\lim_{m \rightarrow \infty}   \displaystyle\inf_{\Prob^m \in \mathcal{P}(\mu,d)^m}  \displaystyle\sup_{p_m \geq 0} \, \frac{\mathrm{REV}(p_m,\Prob^m)}{m} \displaystyle = \mu - d/2.$$ 

It is known that the above relation holds for $m = 1$ \citep[Lemma 1]{chen2023screening}, but this remains unclear when $m > 1$, and, in fact, this presents a relevant open problem. We conjecture that equality does not hold, but it remains unclear how to establish this asymptotically. Do notice that, as the bundling mechanism showing optimality for the maximin problem is deterministic, the equality above also holds if one restricts to deterministic mechanisms. In this setting, it is known already for $m = 1$ that equality does not hold. This follows from the difference in value between the maximin problem studied in \cite{roos2019chebyshev} and the minimax problem in \cite{elmachtoub2021value}.

\subsection{Price alignment through bundling}
Given that deterministic bundling is optimal for maximizing expected revenue in the many-goods limit, we further investigate this subclass of mechanisms for the competitive ratio and the absolute regret objectives. We show that, among all deterministic bundling mechanisms in the many-goods limit, the price that is optimal for expected revenue \eqref{def:price} is also optimal for the competitive ratio and absolute regret.

To proceed, we will formalize these objectives. First, the optimal full-information benchmark is
\begin{align}
\mathrm{OPT}(m,\Prob^m) = \sup_{A' \in \mathcal{A}_m}\mathrm{REV}(A',\Prob^m).
\label{eq:opt}
\end{align}
The corresponding optimal mechanism is generally not known. Regardless, we do know that $\mathrm{OPT}(m,\Prob^m) \leq m \mu$ due to individual rationality. Then, the \textit{robust competitive ratio maximization} problem is defined as
 \begin{align} 
\sup_{p_m}\inf_{\Prob^m \in \mathcal{P}^m}\,   \frac{\mathrm{BUND}(p_m,\Prob^m)}{\mathrm{OPT}(m,\Prob^m)},
\label{eq:robust_ratio_revenue}
\end{align} 
and, analogously, the \textit{robust absolute regret minimization} problem is
 \begin{align}
\inf_{p_m}\sup_{\Prob^m \in \mathcal{P}^m}\, \frac{\mathrm{OPT}(m,\Prob^m) - \mathrm{BUND}(p_m,\Prob^m)}{m}.
\label{eq:robust_regret_revenue}
\end{align}

An interesting open problem is to further investigate the optimality of the bundling price \eqref{def:price} within the general class of randomized mechanisms (or important subsets thereof). This, however, will require a different approach than bounding the min-max problem for competitive ratio, since for any $\bP^m$ (even outside the ambiguity set $\cP(\mu,d)^m$), we obtain
\begin{align*}
    \sup_{A \in \mathcal{A}_m}\frac{\mathrm{REV}(A,\Prob^m)}{\mathrm{OPT}(m,\Prob^m)}=1
\end{align*}
by definition of $\textup{OPT}(m,\bP^m)$ in \eqref{eq:opt}. Consequently, the associated min-max value is always equal to 1. A similar argument can be made for the absolute regret objective.

Theorem \ref{th:uni_price} establishes the universality of deterministic bundling: a single price function is optimal for all three classical objectives in the many-goods limit.

\begin{theorem}[Unified bundling price]\label{th:uni_price}
For $m\xrightarrow{} \infty$, the deterministic bundling mechanism with bundling price 
$$
  p_m^* = (1-\epsilon_m)^2 m\left(\mu - \frac{d}{2(1-\epsilon_m)}\right)
$$
in \eqref{def:price} with $\epsilon_m = \frac{1}{2}(1-\frac{d}{2\mu})m^{-1/5}$, asymptotically attains \eqref{eq:robust_ratio_revenue} and \eqref{eq:robust_regret_revenue}. 
\end{theorem}

We prove Theorem \ref{th:uni_price} in Section \ref{sec:alignment} by solving \eqref{eq:robust_ratio_revenue} and \eqref{eq:robust_regret_revenue} separately for $m \xrightarrow{}\infty$. The first gives a guaranteed competitive ratio of $1 - d/(2\mu)$, while the second gives an absolute regret of at most $d/2$. We then verify that the bundling price \eqref{def:price} is optimal for both objectives. An intuitive way of looking at these results is that the seller can guarantee a revenue close to $m(\mu-d/2)$, whereas the optimal mechanism can achieve a revenue close to $m\mu$, so that the ratio results in $1-d/(2\mu)$ and the absolute regret in $d/2$.

The universality result established in Theorem \ref{th:uni_price} stands in stark contrast to the single-good case. When $m=1$, all three prices for expected revenue, competitive ratio, and absolute regret, denoted by $p^*_{rev}$, $p^*_{cr}$, and $p^*_{ar}$ respectively, are distinct. Theorem \ref{th:3p_m1} shows that these prices are strictly ordered, and therefore different. A similar ordering appears in \cite{chen2022distribution} for the mean-variance ambiguity set.

\begin{proposition}\label{th:3p_m1}
    Let $0 < d < 2\mu$. The optimal robust prices for expected revenue, competitive ratio, and absolute regret satisfy \begin{align}
        p^*_{rev} < p^*_{cr} < p^*_{ar}.
    \end{align}
\end{proposition}

The proof of Proposition \ref{th:3p_m1} is found in Section \ref{sec:single_item}. Figure \ref{fig:3prices} further highlights the contrast with the many-goods setting by showing that $p^*_{ar}$ departs markedly from $p^*_{rev}$ and $p^*_{cr}$. This makes the optimal pricing decision highly sensitive to the choice of objective. In comparison, restricting to deterministic bundling in the many-goods setting resolves this issue.

\begin{figure}[h!]
    \centering
    \includegraphics[width=0.6\linewidth]{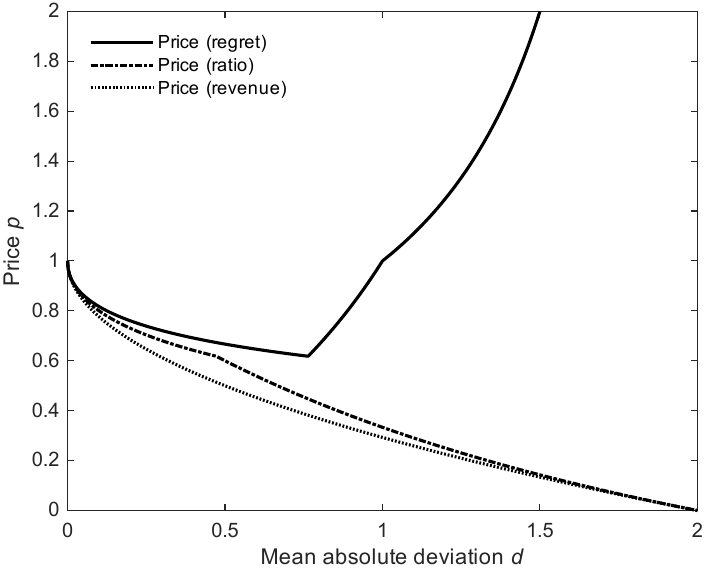}
    \caption{Optimal robust prices for $m=1$ with $\mu=1$.}
    \label{fig:3prices} 
\end{figure}

\section{Proof of the Minimax Revenue Theorem}\label{sec:minimax}
We now start the proof of Theorem \ref{thm:below_mu}. We always have for \textit{any} $A \in \mathcal{A}_m$ that
\begin{align*}
\lim_{m \rightarrow \infty}   \displaystyle \inf_{\Prob^m \in \mathcal{P}(\mu,d)^m} \, \frac{\mathrm{REV}(A,\Prob^m)}{m} \leq \lim_{m \rightarrow \infty}\inf_{\Prob^m \in \mathcal{P}(\mu,d)^m} \sup_{A \in \mathcal{A}_m}  \, \frac{\mathrm{REV}(A,\Prob^m)}{m} \leq \lim_{m \rightarrow \infty}\sup_{A \in \mathcal{A}_m}  \, \frac{\mathrm{REV}(A,\Prob^m_{N})}{m},
\end{align*}
where we choose $\Prob^m_{N} = \Pi_{i=1}^m\mathbb{P}_{N}$ with $\bP_{N}$ as
    \begin{align}\label{def:wcd}
    \bP_{N} = \left\{ \begin{array}{ll}
    x_0 = \mu-\frac{d}{2\alpha(m)} & \text{ w.p. } w_0 = \alpha(m) = 1-(Nm)^{-(m+1)}\cdot e^{-m} \\
    x_n = \mu+\frac{d}{2N}\frac{m^N-1}{(m-1)m^{N-n}(1-\alpha(m))} & \text{ w.p. } w_n = \frac{(m-1)m^{N-n}}{m^N-1}(1-\alpha(m)), \textup{for } n=1,\dots,N
\end{array}\right..
\end{align}
To see that $\bP_{N} \in \cP(\mu,d)$, one can easily verify that $w_n \geq 0$, for all $n \in \{0,1,\dots,N\}$, and that the mean and MAD constraints are satisfied. Also,
\begin{align*}
    &\sum_{n=1}^Nw_n=(1-\alpha(m))\frac{(m-1)}{m^N-1}\sum_{n=1}^N m^{N-n}=(1-\alpha(m))\frac{(m-1)}{m^N-1}\sum_{k=0}^{N-1} m^{k} \\
    &= (1-\alpha(m))\frac{(m-1)}{m^N-1}\frac{m^N-1}{m-1} = 1-\alpha(m),
\end{align*}
so that $w_0 + \sum_{n=1}^Nw_n =1$. For this distribution, every realized valuation for the goods is either $x_0$, or $x_1,\dots,x_N$. We will first argue that, based on the choice of $\alpha(m)$ and the individual rationality property, the only valuation vectors that matter for the revenue analysis, are those with at most one (realized) valuation being $x_1,\dots,x_N$. For a vector $v \in \{x_0,x_1,\dots,x_N\}^m$ we define $C_{\alpha}(v) = |\{v_i : v_i > x_0 \text{ for } i =1,\dots,m\}|$ to be the number of $x_1,\dots,x_N$ in $v$. For any truthful mechanism $A$, we can then write the revenue as 
\begin{align}
\frac{1}{m} \mathrm{REV}(A,\Prob^m_{N}) & = \frac{1}{m}\sum_{v:  C_\alpha(v) \leq 1} \pi(v)\Prob^m_{N}( (X_1,\dots,X_m) = (v_1,\dots,v_m)) \nonumber \\
& + \frac{1}{m}\sum_{k = 2}^m \sum_{v:  C_\alpha(v) = k} \pi(v)\Prob^m_{N}( (X_1,\dots,X_m) = (v_1,\dots,v_m)) \nonumber \\
&  \leq \frac{1}{m}\sum_{v:  C_\alpha(v) \leq 1} \pi(v)\Prob^m_{N}( (X_1,\dots,X_m) = (v_1,\dots,v_m)) \nonumber \\ 
&+  \frac{1}{m} \sum_{k = 2}^m \sum_{v: C_\alpha(v) = k} \left(\sum_{i=1}^m v_i\right)\Prob^m_{N}( (X_1,\dots,X_m) = (v_1,\dots,v_m)).
\label{eq:rev_split}
\end{align}
In the final inequality of \eqref{eq:rev_split}  we use that, because of individual rationality, we always have $\pi(v) \leq \langle z(v),v \rangle \leq \sum_{i=1}^m v_i$. 

We first bound the term $\sum_{v: C_\alpha(v) = k} \left(\sum_{i=1}^m v_i\right)\Prob^m_{N}( (X_1,\dots,X_m) = (v_1,\dots,v_m))$ appearing in \eqref{eq:rev_split} for each $k \in \{2,\dots,m\}$. Fix $k\geq 2$. The set $\{x_0,x_1,\dots,x_N\}^m$ contains exactly $\binom{m}{k}N^k$ valuation vectors $v$ with $C_{\alpha}(v)=k$. Now fix any $v$ such that $C_{\alpha}(v)=k$. We can split $k = \sum_{n=1}^N k_n$ where $k_n$ denotes the number of coordinates $i \in \{1,\dots,m\}$ for which $v_i=x_n$. Define the index set $K = \{n\in\{1,\dots,N\}|k_n> 0\}$. Then, for $m$  enough, we have that
\begin{align}
&\frac{1}{m} \left(\sum_{i=1}^m v_i\right)\Prob^m_{N}( (X_1,\dots,X_m)= (v_1,\dots,v_m)) = \frac{[(m-k)\cdot x_0 + \sum_{n\in K}k_n \cdot x_n]w_0^{m-k}\prod_{n\in K} w_n^{k_n}}{m} \nonumber\\
&\leq [x_0+\sum_{n\in K}x_n]w_0^{m-k}\prod_{n\in K} w_n^{k_n} \leq x_0\prod_{n\in K} w_n^{k_n} + \sum_{n \in K}(x_nw_n)w_n^{k_n-1}\prod_{l\in K\backslash\{n\}}w_l^{k_l}.\label{eq:bound3}
\end{align}
We obtain this bound by using $(m-k)/m \leq 1$, $k_n/m \leq 1$, and $w_0=\alpha(m)\leq 1$. Next, for all $n \in K$, notice that $w_n^{k_n-1}\prod_{l\in K\backslash\{n\}}w_l^{k_l} \leq (1-\alpha(m))^{k_n-1}\prod_{l\in K\backslash\{n\}}(1-\alpha(m))^{k_l} \leq (1-\alpha(m))^{k-1} \leq 1-\alpha(m)$. Similarly, $\prod_{n \in K}w_n^{k_n} \leq 1-\alpha(m)$. Hence, \eqref{eq:bound3} can be further bounded to
\begin{align}
&\frac{1}{m} \left(\sum_{i=1}^m v_i\right)\Prob^m_{N}( (X_1,\dots,X_m)= (v_1,\dots,v_m)) \leq [x_0+ \sum_{n \in K}x_nw_n](1-\alpha(m))\nonumber \\
&\leq [x_0+ \sum_{n = 1}^Nx_nw_n](1-\alpha(m)) \leq 2\mu \cdot (Nm)^{-(m+1)} \cdot e^{-m}.
\label{eq:rev_split_second}
\end{align}
The last inequality follows by observing that $x_nw_n\leq \mu + d/(2N)$. This implies that, still for $k \geq 2$, 
$$
\frac{1}{m} \sum_{v:  C_\alpha(v) = k} \left(\sum_{i=1}^m v_i\right)\Prob^m_{N}( (X_1,\dots,X_m) = (v_1,\dots,v_m)) \leq \binom{m}{k}N^k\cdot 2\mu \cdot (Nm)^{-(m+1)} \cdot e^{-m} \leq \frac{2\mu}{N}\cdot \frac{e^{-m}}{m}
$$
using that $|v:  C_\alpha(v) = k| = \binom{m}{k}N^k \leq m^kN^k\leq (Nm)^m$. Finally, because $k \in \{2,\dots,m\}$ in the outer summation, i.e., there are $m-1$ terms, we obtain
\begin{align}
\frac{1}{m} \sum_{k = 2}^m \sum_{v:  C_\alpha(v) = k} \left(\sum_{i=1}^m v_i\right)\Prob^m_{N}( (X_1,\dots,X_m) = (v_1,\dots,v_m)) \leq (m-1)\frac{2\mu}{N}\cdot \frac{e^{-m}}{m} \leq \frac{2\mu}{N} \cdot  e^{-m}.
\label{eq:rev_split_part2}
\end{align}
Taking $m \rightarrow \infty$ in the right hand side of \eqref{eq:rev_split_part2} gives an upper bound of $0$ on the contribution of all these valuation vectors $v$ to the revenue in \eqref {eq:rev_split}.

We next continue with bounding the summation
\begin{align}
\frac{1}{m}\sum_{v:  C_{\alpha}(v) \leq 1} \pi(v)\Prob^m_{N}( (X_1,\dots,X_m) = (v_1,\dots,v_m)) = \frac{1}{m} \pi\left(v^{(0)}\right)w_0^m  + \sum_{n=1}^N\frac{1}{m}\sum_{j=1}^m \pi\left(v^{(n,j)}\right)w_0^{m-1}w_n,
\label{eq:rev_split_part1a}
\end{align}
where $v^{(n,j)}$ is such that $v^{(n,j)}_k = \left\{ \begin{array}{ll}
 x_n & \text{ for } k = j \\
 x_0 & \text{ for } k \neq j 
\end{array}\right.$ and $v^{(0)} = (x_0,x_0,\dots,x_0)$ the vector only containing $x_0$'s. We denote (averaged) payments as $\pi_0 =  \pi\left(v^{(0)}\right)$ and $\bar{\pi}_n = \frac{1}{m}\sum_{j=1}^m\pi(v^{(n,j)})$ for $n \in \{1,\dots,N\}$. We also define the average allocation probabilities as $\bar{s}_0 = \frac{1}{m}\sum_{j=1}^m z_j(v^{(0)})$ and $\bar{s}_n = \frac{1}{m}\sum_{j=1}^m z_j(v^{(n,j)})$ for $n \in \{1,\dots,N\}$. Using these definitions, we can write \eqref{eq:rev_split_part1a} as
\begin{align}
    \frac{1}{m} \pi_0w_0^m  + \sum_{n=1}^N\bar{\pi}_nw_0^{m-1}w_n\leq \frac{1}{m}\pi_0  + \sum_{n=1}^N\bar{\pi}_nw_n.
\end{align}
By individual rationality, we have that
\begin{align}
\pi_0 \leq \langle z(v^{(0)}), v^{(0)} \rangle = mx_0\bar{s}_0.
\label{eq:ind_rat_v0}
\end{align}

Next, we use incentive compatibility regarding $v^{(n,j)}$ and $v^{(n-1,j)}$ for all $n \in \{1,\dots,N\}$. This results in
\begin{align}
    \langle z(v^{(n,j)}),v^{(n,j)} \rangle - \pi(v^{(n,j)}) \geq \langle z(v^{(n-1,j)}),v^{(n,j)} \rangle - \pi(v^{(n-1,j)}), \ \ j = 1,\dots,m.
\end{align}
We now expand both products. First, expanding the inner product on the left-hand side results in 
\begin{align}
    \langle z(v^{(n,j)}),v^{(n,j)} \rangle = \sum_{i\neq j}[x_0z_i(v^{(n,j)})] + x_nz_j(v^{(n,j)}) = \sum_{i=1}^m[x_0z_i(v^{(n,j)})] + (x_n-x_0)z_j(v^{(n,j)}).
\end{align}
Then, expanding the inner product on the right-hand side results in 
\begin{align}
    &\langle z(v^{(n-1,j)}),v^{(n,j)} \rangle = \sum_{i\neq j}[x_0z_i(v^{(n-1,j)})] + x_n z_j(v^{(n-1,j)}) \nonumber \\
    &= \sum_{i=1}^m[x_0z_i(v^{(n-1,j)})] + (x_n-x_0)z_j(v^{(n-1,j)}).
\end{align}
Hence, for all $j = 1,\dots,m$, we obtain
\begin{align}\label{eq:IC}
     &\pi(v^{(n,j)}) \leq \pi(v^{(n-1,j)}) + (x_n-x_0)(z_j(v^{(n,j)}) - z_j(v^{(n-1,j)})) + x_0\left[\sum_{i=1}^mz_i(v^{(n,j)})-\sum_{i=1}^mz_i(v^{(n-1,j)})\right] \nonumber \\
     &\leq \pi(v^{(n-1,j)}) + (x_n-x_0)(z_j(v^{(n,j)}) - z_j(v^{(n-1,j)})) + x_0m.
\end{align}
Define $\Delta_n := \bar{s}_n - \bar{s}_{n-1}$ for all $n \in \{1,\dots,N\}$. Summing \eqref{eq:IC} for each $j$ and dividing by $m$ gives
\begin{align}\label{eq:pibounds}
     &\bar{\pi}_n \leq \bar{\pi}_{n-1} + (x_n-x_0)(\bar{s}_n - \bar{s}_{n-1}) + x_0m.
\end{align}
We then obtain, recursively, by combining \eqref{eq:pibounds} and \eqref{eq:ind_rat_v0}, the upper bound
\begin{align}\label{eq:pibounds2}
     &\bar{\pi}_n \leq nx_0m + mx_0\bar{s}_0 + \sum_{k=1}^n(x_k-x_{0})\Delta_k \leq (N+1)m\left(\mu-\frac{d}{2}\right)+\sum_{k=1}^n(x_k-x_{0})\Delta_k,
\end{align}
and consequently,
\begin{align}
    \sum_{n=1}^N\bar{\pi}_nw_n \leq N(N+1)m\left(\mu-\frac{d}{2}\right)(1-\alpha(m)) + \sum_{n=1}^Nw_n\sum_{k=1}^n(x_k-x_{0})\Delta_k  \leq C(m) + \sum_{n=1}^N x_nw_n\Delta_n
\end{align}
with $C(m)$ such that $\lim_{m\xrightarrow{}\infty}C(m) = 0$. This follows because, for any $n \in \{1,\dots,N\}$, we have $\lim_{m\xrightarrow{}\infty}w_nx_k=\lim_{m\xrightarrow{}\infty}\left(w_n\mu+\frac{d}{2N}\cdot m^{-(n-k)}\right) =0$, for all $k < n$. 
We now focus on bounding $\sum_{n=1}^N x_nw_n\Delta_n$. We start by giving bounds on $\Delta_n$.
\begin{lemma}\label{lemma:monotonicity}
    Let $n \in \{1,\dots,N\}$. The following holds: 
    $$\Delta_n \geq 0.$$
\end{lemma}
\begin{proof}
    Since $A\in\mathcal{A}_m$ is truthful, we have due to incentive compatibility that
    \begin{align*}
        \langle z(v^{(n,j)}),v^{(n,j)} \rangle - \pi(v^{(n,j)}) \geq \langle z(v^{(n-1,j)}),v^{(n,j)} \rangle - \pi(v^{(n-1,j)}),
    \end{align*}
    and 
    \begin{align*}
        \langle z(v^{(n-1,j)}),v^{(n-1,j)} \rangle - \pi(v^{(n-1,j)}) \geq \langle z(v^{(n,j)}),v^{(n-1,j)} \rangle - \pi(v^{(n,j)}).
    \end{align*}
    Adding the incentive compatibility constraints results in
    \begin{align*}
        \langle z(v^{(n,j)}),v^{(n,j)} \rangle+\langle z(v^{(n-1,j)}),v^{(n-1,j)} \rangle \geq \langle z(v^{(n-1,j)}),v^{(n,j)} \rangle + \langle z(v^{(n,j)}),v^{(n-1,j)} \rangle.
    \end{align*}
    By rearranging terms we get
    \begin{align*}
        \langle z(v^{(n,j)}),v^{(n,j)} \rangle - \langle z(v^{(n,j)}),v^{(n-1,j)} \rangle \geq \langle z(v^{(n-1,j)}),v^{(n,j)} \rangle - \langle z(v^{(n-1,j)}),v^{(n-1,j)} \rangle,
    \end{align*}
    or equivalently,
    \begin{align*}
        \langle z(v^{(n,j)}),v^{(n,j)} - v^{(n-1,j)} \rangle \geq \langle z(v^{(n-1,j)}),v^{(n,j)}-v^{(n-1,j)} \rangle.
    \end{align*}
    This implies that
    \begin{align*}
        \langle z(v^{(n,j)})-z(v^{(n-1,j)}),v^{(n,j)} - v^{(n-1,j)} \rangle \geq 0.
    \end{align*}
    Since the valuations $v^{(n,j)},v^{(n-1,j)}$ are the same except at location $j$, this results in
    \begin{align*}
        (x_n - x_{n-1}) \cdot (z_j(v^{(n,j)}) - z_j(v^{(n-1,j)})) \geq 0.
    \end{align*}
    Since $x_n > x_{n-1}$, we have that $z_j(v^{(n,j)}) \geq z_j(v^{(n-1,j)})$. Taking the average (summing over all $j$ and dividing by $m$) gives $\bar{s}_n \geq \bar{s}_{n-1}$, which implies $\Delta_n \geq 0$.
\end{proof}

Next, notice that $\sum_{n=1}^N\Delta_n = \bar{s}_N-\bar{s}_0 \leq 1-\bar{s}_0$. Then, in combination with Lemma \ref{lemma:monotonicity}, we get
\begin{align}
    \sum_{i=1}^Nx_nw_n\Delta_n \leq \max\{w_1x_1,\dots,w_Nx_N\}\sum_{n=1}^N\Delta_n \leq \max\{w_1x_1,\dots,w_Nx_N\}(1-\bar{s}_0).
 \end{align}
Taking $N = \lceil\frac{d}{2(\mu-d/2)}\rceil$ and $m \xrightarrow{} \infty$ then gives
\begin{align}
    &\lim_{m\xrightarrow{}\infty} \frac{1}{m}\pi_0  + \sum_{n=1}^N\bar{\pi}_nw_n. \leq \lim_{m\xrightarrow{}\infty} x_0\bar{s}_0 + C(m) + \max\{w_1x_1,\dots,w_Nx_N\}(1-\bar{s}_0)\nonumber \\
    &\leq \left(\mu-\frac{d}{2}\right)\bar{s}_0 + \frac{d}{2N}(1-\bar{s}_0) \leq \max\left\{\mu-\frac{d}{2},\frac{d}{2N}\right\} = \mu - \frac{d}{2}.
\end{align}
This completes the proof.

\section{Proof of the Maximin Bundling Theorem}\label{sec:maximin}
In this section we prove the result of Theorem \ref{thm:mu_d_2mu}. First, we derive a one-sided concentration bound for fixed mean and MAD, and then we use this bound to complete the proof.

\subsection{Deriving the one-sided concentration bound}
\label{sec:concentration}
The cornerstone of the proof of Theorem \ref{thm:mu_d_2mu} is Proposition \ref{prop:tail_sum}, which might be of independent interest. Informally, it states that the sum of $m$ i.i.d. random variables with mean $\mu$ and mean absolute deviation (MAD) $d$ will be greater than or equal to $\approx m(\mu - d/2)$ with probability close to one as $m$ grows large. We will use the following one-sided Chebyshev bound in the proof of Proposition \ref{prop:tail_sum}.

\begin{lemma}[Chebyshev bound]
If $Y = \sum_{i=1}^m X_i$ is the sum of i.i.d. random variables $X_1,\dots,X_m$ that all have mean $\mu$ and variance at most $\sigma^2$, then for any $0 < \gamma < 1$ it holds that
\begin{align}
    \Prob(Y \leq (1-\gamma)m\mu) \leq \frac{m \sigma^2}{\gamma^2m^2\mu^2} = \frac{\sigma^2}{(\gamma\mu)^2m}.
\end{align}
\label{lem:chebyshev}
\end{lemma}

We now state the one-sided concentration result.
\begin{proposition}[One-sided concentration bound]
Let $\mu > 0$, $0 \leq d < 2\mu$, and $\epsilon \in (0,\frac{1}{2}(1-\frac{d}{2\mu}))$ be given. Then, for every  $m \in \N$, it holds that
\begin{align}
\inf_{\Prob^m \in \cP(\mu,d)^m} \Prob^m\left(\sum_{i=1}^m X_i \geq (1-\epsilon)^2 m\left(\mu - \frac{d}{2(1-\epsilon)}\right)\right) \geq 1 - \frac{f(\mu,d)}{m\epsilon^4},
\label{eq:concentration}
\end{align}
where $f(\mu,d)=\left(\frac{2\mu+d}{2\mu-d}\right)^2$ is a function independent of $m$ and $\epsilon$.
\label{prop:tail_sum}
\end{proposition}
\begin{proof}
The outline of the proof is as follows. We show that, if one chooses a fixed $t = t(\mu,d,\epsilon)$ large enough, then the contribution of the values greater than or equal to $t$ to the mean $\mu$ of any $\mathbb{P}_i \in \mathcal{P}(\mu,d)$ is roughly at most $g(t) \cdot \mu + d/2$ with $g(t)$ small if $t$ is large enough. If we condition the distribution $\mathbb{P}_i$ on values smaller than $t$, then the resulting distribution has mean $(1 - g(t))\mu - d/2$ and finite variance. We can then use Chebyshev's bound to argue that we have concentration around $m((1 - g(t))\mu - d/2)$. This will give the desired result.

We continue with the formal proof, for which we will use the following lemma. 
\begin{lemma}
Let $t \geq \mu + d/2$ be fixed. Then
\begin{align}
\sup_{\mathbb{P} \in \mathcal{P}(\mu,d)} \mathbb{E}_{\mathbb{P}}[X \cdot \1_{\geq t}(X)] =  \frac{d}{2(t-\mu)}\cdot\mu + \frac{d}{2},
\end{align}
where $\1_{\geq t}(x) = 1$ if $x \geq t$ and $\1_{\geq t}(x) = 0$ if $x < t$.
\label{prop:tails}
\end{lemma}
\begin{proof}
Let $\mathbb{P} \in \Pmud$ be a fixed probability distribution. We can merge all the probability mass under $t$ in one point, as well as all the probability mass above $t$. Formally speaking, we look at the two-point distribution $\mathbb{P}'$ supported on $\{x',y'\}$ with $x' = \mathbb{E}[X | X < t]$, $y' = \mathbb{E}[X | X \geq t]$ and $\alpha' = \mathbb{P}(X < t)$ the probability mass on point $x'$. It is not hard to see that distribution $\mathbb{P}'$ has mean $\mu$ and a MAD $d' \leq d$, because of the convex nature of the function $\phi(x) = |x - \mu|$. Furthermore, it holds that $\mathbb{E}_{\mathbb{P}}[X \cdot \1_{\geq t}(X)] = \mathbb{E}_{\mathbb{P}'}[X \cdot \1_{\geq t}(X)]$ since the function $x \cdot \1_{\geq t}(x)$ is piece-wise linear (and the fact that we reduce to a two-point distribution using the point $t$).

Now, we have $\mathbb{E}_{\mathbb{P}'}[X \cdot \1_{\geq t}(X)] = y'(\alpha')(1-\alpha') = (1-\alpha')\mu + d'/2$ if $y'(\alpha') \geq t$ (note that also $t \geq \mu + d'/2$), and zero otherwise. In order to maximize this quantity, we should choose $1 - \alpha$ as large as possible, but still small enough so that $y(\alpha) \geq t$ (otherwise the objective equals zero). The latter inequality is equivalent to $1-\alpha \leq d'/(2(t-\mu))$, so we choose $1 - \alpha' = d'/(2(t-\mu))$. It then follows that
$$
\sup_{\mathbb{P} \in \mathcal{P}(\mu,d)} \mathbb{E}_{\mathbb{P}}[X \cdot \1_{\geq t}(X)] \leq \frac{d'}{2(t-\mu)}\cdot\mu + \frac{d'}{2} \leq \frac{d}{2(t-\mu)}\cdot\mu + \frac{d}{2},
$$
where the second inequality holds because $d' \leq d$. The bound can be attained by choosing the two-point distribution supported on $\{x(\alpha),y(\alpha)\}$ for which $1 - \alpha = d/(2(t-\mu))$. The proof of this claim is the same as the reasoning we gave above for the distribution supported on $\{x',y'\}$.
\end{proof}
Now fix an arbitrary $\epsilon \in (0,\frac{1}{2}(1 - \frac{d}{2\mu}))$ and choose $t = \mu+d/(2\epsilon)$. Because $\epsilon \in (0,\frac{1}{2}(1-\frac{d}{2\mu})) \subset (0,1)$, this will always result in $t \geq \mu+d/2$. For any given $\Prob^m \in \cP(\mu,d)^m$, let $\Prob^m_t(h(X_1,\dots,X_m) \in E) := \Prob^m(h(X_1,\dots,X_m) \in E | X_1 < t,\dots,X_m<t)$ for any function $h$ and event $E$. Note that by definition of $\Prob^m_t$, for a given $c_m \geq 0$ it holds that
\begin{align}
\Prob^m\left(\sum_{i=1}^m X_i \geq c_m\right) \geq \Prob^m\left(\sum_{i=1}^m X_i \geq c_m \ \Big| \ X_i < t \ \forall i\right) = \Prob_t^m\left(\sum_{i=1}^m X_i^t \geq c_m\right)
\label{eq:estimate_cond}
\end{align}
with $(X_1,\dots,X_m) \sim \Prob^m$ and $(X^t_1,\dots,X^t_m) \sim \Prob^m_t$. Because of Lemma \ref{prop:tails}, and using $\Prob_i(X_i < t) \leq 1$, we have
\begin{align}
    \mu^t_i := \mathbb{E}_{\Prob_i}[X_i | X_i < t]= \frac{\mu - \mathbb{E}_{\Prob_i}[X_i \cdot \1_{\geq t}(X_i)]}{\Prob_i(X_i < t)}  \geq  (1- \epsilon)\mu - \frac{d}{2},
\label{eq:cond_mean}
\end{align}
and define $\mu^t_{\textup{min}} := \min\{\mu^t_1,\dots,\mu^t_m\}$.
The conditional distribution $\Prob_i(X_i |X_i < t)$ has finite support on $[0,t)$ and therefore has
\begin{align}
    (\sigma^t_i)^2 := \mathbb{E}_{\Prob_i}\left((X^t_i-\mu^t_i)^2|X_i < t\right) \leq \mu^t_i(t-\mu^t_i),
\end{align}
which is maximal when $\mu^t_i = \frac{1}{2}t$, as $\mu^t_i \in [0,t)$. Hence, we obtain the upper bound
\begin{align}\label{sigma^t}
    (\sigma^t_i)^2 \leq \frac{1}{2}t(t-\frac{1}{2}t) = \frac{1}{4}t^2.
\end{align}
Then Chebyshev's bound in Lemma \ref{lem:chebyshev} with \eqref{sigma^t} tells us that
\begin{align}
\Prob_t^m\left(\sum_{i=1}^m X_i^t > (1- \epsilon)m\mu^t_{\textup{min}}\right) \geq 1 - \frac{t^2}{4(\epsilon \mu^t_{\textup{min}})^2m},
\end{align}
Combined with \eqref{eq:estimate_cond} and \eqref{eq:cond_mean}, this then tells us that
\begin{align}\label{eq:function_f}
&\Prob^m\left(\sum_{i=1}^m X_i > (1- \epsilon)m\left((1- \epsilon)\mu - \frac{d}{2}\right)\right) \geq 1 - \frac{t^2}{\epsilon^2(2(1-\epsilon)\mu-d)^2m}\nonumber \\
& = 1 - \frac{(2\mu \epsilon+d)^2}{4\epsilon^4(2(1-\epsilon)\mu-d)^2m}.
\end{align}
Since $2\mu\epsilon \leq (2\mu-d)/2$ and $2(1-\epsilon)\mu-d\geq (2\mu-d)/2$, 
\eqref{eq:function_f} results in
\begin{align}
    \Prob^m\left(\sum_{i=1}^m X_i > (1- \epsilon)m\left((1- \epsilon)\mu - \frac{d}{2}\right)\right) \geq 1 - \frac{(2\mu+d)^2}{(2\mu-d)^2m\epsilon^4}.
\end{align}
As $\Prob^m$ was chosen arbitrarily, this inequality also holds for the infimum over $\mathbb{P}^m$. 
\end{proof}

\subsection{Applying the one-sided concentration result}\label{sec:3.1}
Recall from \eqref{eq:rev_bundling} that $\mathrm{BUND}(p_m,\Prob^m) = p_m \Prob^m(\sum_{i=1}^m X_i \geq p_m)$. Using Proposition \ref{prop:tail_sum}, it follows that for every $0 < \epsilon < \frac{1}{2}(1 - d/2\mu)$, we have
\begin{align}
\sup_{p_m} \inf_{\Prob^m \in \mathcal{P}(\mu,d)^m} \, \frac{p_m \cdot \Prob^m(\sum_{i=1}^m X_i \geq p_m)}{m} &\geq  \inf_{\Prob^m \in \mathcal{P}(\mu,d)^m} \, \frac{p_m^* \cdot \Prob^m(\sum_{i=1}^m X_i \geq p_m^*)}{m}\nonumber\\
&\geq \frac{p_m^*}{m}\left(1 - \frac{f(\mu,d)}{m\epsilon^4}\right).
\label{eq:lower_maximin}
\end{align}
Take $\epsilon = \epsilon_m = \frac{1}{2}(1-\frac{d}{2\mu})m^{-1/5}$, so that the expression on the right approaches $p_m^*/m$ as $m \rightarrow \infty$. Note that any $\epsilon$ of the form $\frac{1}{2}(1-\frac{d}{2\mu})m^{-1/(4+\delta)}$ with $\delta > 1$ is sufficient.

On the other hand, if $p_m \geq m(\mu - d/2)$, then nature can asymptotically put all probability mass of $\sum_{i=1}^m X_i$ below $p_m$ by choosing the two-point distribution $\mathbb{P}_{\alpha}$ with $\alpha \rightarrow 1$ in \eqref{eq:two_point_mean_mad} for all $i \in \{1,\dots,m\}$. That is, the leftmost point in the support of $\sum_{i=1}^m X_i$, $m \cdot x(\alpha)$ approaches $m(\mu - d/2)$ from the left, and the mass on this point approaches $1$, as follows from \eqref{eq:sum_Xi}. This implies that $\inf_{\Prob^m \in \mathcal{P}(\mu,d)^m} \, p_m \cdot \Prob^m(\sum_{i=1}^m X_i \geq p_m) / m = 0,$
when $p_m \geq m(\mu - d/2)$. Therefore 
\begin{align}
\sup_{p_m} \inf_{\Prob^m \in \mathcal{P}(\mu,d)^m} \, \frac{p_m \cdot \Prob^m(\sum_{i=1}^m X_i \geq p_m)}{m} = \sup_{p_m < m(\mu-d/2)} \inf_{\Prob^m \in \mathcal{P}(\mu,d)^m} \, \frac{p_m \cdot \Prob^m(\sum_{i=1}^m X_i \geq p_m)}{m} \leq \mu-\frac{d}{2}
\label{eq:upper_maximin}
\end{align}
using that $\Prob^m(\sum_{i=1}^m X_i \geq p_m) \leq 1$. Combining the lower and upper bounds in \eqref{eq:lower_maximin} and \eqref{eq:upper_maximin}, respectively, it follows that
$$
\mu-\frac{d}{2}=\lim_{m\xrightarrow{}\infty}(1-\epsilon_m)^2\left(\mu - \frac{d}{2(1-\epsilon_m)}\right) \leq \lim_{m \rightarrow \infty} \sup_{p_m} \inf_{\Prob^m \in \mathcal{P}(\mu,d)^m} \,  \frac{ \mathrm{BUND}(p_m,\Prob^m)}{m}  \leq \mu - \frac{d}{2},
$$
since $\lim_{m\xrightarrow{}\infty}\epsilon_m= \lim_{m\xrightarrow{}\infty}\frac{1}{2}(1-\frac{d}{2\mu})m^{-1/5}=0$. This completes the proof.

\section{Unified bundling price in the many-goods limit}\label{sec:alignment}
We will prove Theorem \ref{th:uni_price} by resolving \eqref{eq:robust_ratio_revenue} and \eqref{eq:robust_regret_revenue} separately while demonstrating that $p^*_m$ achieves the optimal values.

\subsection{Competitive ratio under bundling}
We first analyze the competitive ratio by solving \eqref{eq:robust_ratio_revenue}.
\begin{proposition} 
Let $0 \leq d < 2\mu$. The maximin competitive ratio satisfies
\begin{align}\label{eq:ratio_m}
\lim_{m\rightarrow \infty} \sup_{p_m} \inf_{\Prob^m \in \cP(\mu,d)^m} \frac{\mathrm{BUND}(p_m,\Prob^m)}{\mathrm{OPT}(m,\Prob^m)} =  1 - \frac{d}{2\mu},
\end{align}
where the supremum is asymptotically attained by bundling price $p^*_m$ in \eqref{def:price}.
\label{thm:maximin_ratio}
\end{proposition}
\begin{proof}
We will prove \eqref{eq:ratio_m} by proving the following inequalities:
\begin{align}
1 - \frac{d}{2\mu} &\leq \lim_{m\rightarrow \infty}  \left[\sup_{p_m} \inf_{\Prob^m \in \cP(\mu,d)^m} \frac{\mathrm{BUND}(p_m,\Prob^m)}{\mathrm{OPT}(m,\Prob^m)} \right] \nonumber\\
&\leq
\lim_{m\rightarrow \infty} \left[\sup_{p_m} \inf_{\Prob^m \in \cP_2(\mu,d)^m} \frac{\mathrm{BUND}(p_m,\Prob^m)}{\mathrm{OPT}(m,\Prob^m)} \right] \leq 1 - \frac{d}{2\mu}.
\label{eq:inequalities}
\end{align}
The second inequality trivially holds, as we take the infimum over a larger set on the left-hand side compared to the right-hand side. We start by proving the first inequality in \eqref{eq:inequalities}.

Consider $p_m^*$. We have
\begin{align}
\sup_{p_m} \inf_{\Prob^m \in \cP(\mu,d)^m} \frac{\mathrm{BUND}(p_m,\Prob^m)}{\mathrm{OPT}(m,\Prob^m)} \geq  \inf_{\Prob^m \in \cP(\mu,d)^m} \frac{\mathrm{BUND}(p_m^*,\Prob^m)}{\mathrm{OPT}(m,\Prob^m)} \geq \frac{\inf_{\Prob^m \in \cP(\mu,d)^m} \mathrm{BUND}(p_m^*,\Prob^m)}{m\mu}
\label{eq:bundle_tight2}
\end{align}
using that $\mathrm{OPT}(m,\Prob^m) \leq m \mu$. This inequality holds because any truthful mechanism is individually rational, which implies that $\pi(v) \leq \sum_i v_i$. In expectation, this yields a bound of $m\mu$. Then, plugging in the concentration bound \eqref{eq:concentration} with $\epsilon = \epsilon_m = \frac{1}{2}(1-\frac{d}{2\mu})m^{-1/5}$ in \eqref{eq:bundle_tight2} gives 
$$
\sup_{p_m} \inf_{\Prob^m \in \cP(\mu,d)^m} \frac{\mathrm{BUND}(p_m,\Prob^m)}{\mathrm{OPT}(m,\Prob^m)} \geq \frac{(1-\epsilon_m)^2 m\left(\mu - \frac{d}{2(1-\epsilon_m)}\right)(1 - f(\mu,d)/(m\epsilon_m^4))}{m\mu},
$$
and then 
$$
\lim_{m \rightarrow \infty} \left[ \sup_{p_m} \inf_{\Prob^m \in \cP(\mu,d)^m} \frac{\mathrm{BUND}(p_m,\Prob^m)}{\mathrm{OPT}(m,\Prob^m)}\right] \geq \frac{\mu-d/2}{\mu} = 1-\frac{d}{2\mu}.
$$

For any fixed $m$, it is not hard to argue that for some function $g(\mu,d) > 0$ independent of $m$, we have $\mathrm{OPT}(m,\Prob^m) \geq g(\mu,d) > 0$ for every $\Prob^m \in \cP_2(\mu,d)^m$, e.g., by using the mechanism that separately sells every good using the optimal robust price in the single good setting. Consider the following case distinction.

\textbf{Case 1: $p_m \geq m(\mu - d/2)$.} Nature can asymptotically get all probability mass of $Y=\sum_{i=1}^m X_i$ below $p_m$, by choosing $\Prob^m_{\alpha}=\Pi_{i=1}^m\mathbb{P}_{\alpha}$ with $\mathbb{P}_{\alpha}$ the two-point distribution with  $\alpha \rightarrow 1$ in \eqref{eq:two_point_mean_mad}, for all $i \in \{1,\dots,m\}$. That is, the leftmost point in the support of $Y$, $m \cdot x(\alpha)$ approaches $m(\mu - d/2)$ from the left and the mass on this point approaches $1$, as follows from \eqref{eq:sum_Xi}. This implies that the revenue of bundling approaches zero as $\alpha \rightarrow 1$ (recall that $m$ is fixed at this point), and so 
\begin{align}
\inf_{\Prob^m \in \mathcal{P}_2(\mu,d)^m} \frac{\mathrm{BUND}(p_m,\Prob^m)}{\mathrm{OPT}(m,\Prob^m)}  \leq \inf_{\Prob^m \in \mathcal{P}_2(\mu,d)^m} \frac{\mathrm{BUND}(p_m,\Prob^m)}{g(\mu,d)}=  \frac{\inf_{\Prob^m \in \mathcal{P}_2(\mu,d)^m} \mathrm{BUND}(p_m,\Prob^m)}{g(\mu,d)}  = 0.
\label{eq:case1}
\end{align}

\textbf{Case 2: $p_m < m(\mu-d/2)$.} Note that always
$$
\mathrm{BUND}(p_m,\Prob^m) = p_m \Prob^m\left(\sum_{i=1}^m X_i \geq p_m\right) \leq p_m.
$$
This implies that
\begin{align}
\inf_{\Prob^m \in \cP_2(\mu,d)^m} \frac{\mathrm{BUND}(p_m,\Prob^m)}{\mathrm{OPT}(m,\Prob^m)} \leq \inf_{\Prob^m \in \cP_2(\mu,d)^m} \frac{p_m}{\mathrm{OPT}(m,\Prob^m)} =  \frac{p_m} {\sup_{\Prob^m \in \cP_2(\mu,d)^m}\mathrm{OPT}(m,\Prob^m)},
\label{eq:case2}
\end{align}
where the last equality holds because $p_m$ does not depend on $\Prob^m$. 

Fix any $\gamma > 0$. Using the joint distribution $\Prob^m_{\alpha} = \Pi_{i=1}^m \mathbb{P}_{\alpha}$ with $\mathbb{P}_{\alpha}$ the fixed two-point distribution with $\alpha = d/(2\mu)$, and using the bundling mechanism that sells at the bundle price $(1-\gamma)m\mu$, gives
\begin{align}
\sup_{\Prob^m \in \cP_2(\mu,d)^m}\mathrm{OPT}(m,\Prob^m) \geq \mathrm{BUND}((1-\gamma )m\mu,\Prob^m_{\alpha}) \geq (1-\gamma )m\mu \cdot \left[1 - \frac{\mathrm{Var}(X)}{(\gamma \mu)^2 m}\right]
\label{eq:case2_2}
\end{align}
by using Lemma \ref{lem:chebyshev}. Note that, when $X \sim \mathbb{P}_{\alpha}$,
$$
\text{Var}(X) = (d/2\mu)\cdot (0 - \mu)^2 + (1 - d/(2\mu))\cdot (\mu - \mu - d\mu/(2\mu - d))^2 =: h(\mu,d).
$$
Plugging \eqref{eq:case2_2} into \eqref{eq:case2} gives that for $p_m < m(\mu - d/2)$, we have
\begin{align}
\inf_{\Prob^m \in \cP_2(\mu,d)^m} \frac{\mathrm{BUND}(p_m,\Prob^m)}{\mathrm{OPT}(m,\Prob^m)}   & \leq \frac{p_m}{(1-\gamma )m\mu \cdot \left[1 - \frac{h(\mu,d)}{(\gamma\mu)^2 m}\right]} \leq \frac{m(\mu-d/2)}{(1-\gamma)m\mu \cdot \left[1 - \frac{h(\mu,d)}{(\gamma\mu)^2 m}\right]} \nonumber \\
& = \frac{1}{1-\gamma} \cdot \frac{2\mu-d}{2\mu} \cdot \left[1 - \frac{h(\mu,d)}{(\gamma\mu)^2 m}\right]^{-1}.
\label{eq:case2_3}
\end{align}
Combining Cases 1 and 2 implies that 
$$
\sup_{p_m} \inf_{\Prob^m \in \cP_2(\mu,d)^m} \frac{\mathrm{BUND}(p_m,\Prob^m)}{\mathrm{OPT}(m,\Prob^m)} \leq  \frac{1}{1-\gamma} \cdot \left(1 - \frac{d}{2\mu}\right) \cdot \left[1 - \frac{h(\mu,d)}{(\gamma\mu)^2 m}\right]^{-1},
$$
and so
$$
\lim_{m \rightarrow \infty} \left[\sup_{p_m} \inf_{\Prob^m \in \cP_2(\mu,d)^m} \frac{\mathrm{BUND}(p_m,\Prob^m)}{\mathrm{OPT}(m,\Prob^m)}\right] \leq  \frac{1}{1-\gamma} \cdot \left(1 - \frac{d}{2\mu}\right).
$$
Since $\gamma > 0$ was chosen arbitrarily, we obtain
\begin{align}
\lim_{m \rightarrow \infty} \left[\sup_{p_m} \inf_{\Prob^m \in \cP_2(\mu,d)^m} \frac{\mathrm{BUND}(p_m,\Prob^m)}{\mathrm{OPT}(m,\Prob^m)}\right] \leq 1 - \frac{d}{2\mu}.
\label{eq:lower_bound}
\end{align}
This finishes the proof.
\end{proof}

\subsection{Absolute regret under bundling}
We now turn to the analysis of absolute regret by proving \eqref{eq:robust_regret_revenue}.

\begin{proposition}
Let $0 \leq d < 2\mu$. The minimax absolute regret satisfies
\begin{align}\label{eq:regret_m}
\lim_{m\rightarrow \infty} \inf_{p_m} \sup_{\Prob^m \in \mathcal{P}(\mu,d)^m} \frac{\mathrm{OPT}(m,\Prob^m) - \mathrm{BUND}(p_m,\Prob^m)}{m}  =  \frac{d}{2}, \, 
\end{align}
where the infimum is asymptotically attained by bundling price $p^*_m$ in \eqref{def:price}.
\label{thm:minimax_regret}
\end{proposition}
\begin{proof}
First, suppose that the seller chooses $p_m^*$ as the selling price. Then
\begin{align*}
&\inf_{p_m} \sup_{\Prob^m \in \mathcal{P}(\mu,d)^m} \frac{\mathrm{OPT}(m,\Prob^m) - \mathrm{BUND}(p_m,\Prob^m)}{m} \leq \sup_{\Prob^m \in \mathcal{P}(\mu,d)^m} \frac{\mathrm{OPT}(m,\Prob^m) - \mathrm{BUND}(p_m^*,\Prob^m)}{m} \\
& \leq \sup_{\Prob^m \in \mathcal{P}(\mu,d)^m}\frac{\mathrm{OPT}(m,\Prob^m)}{m} -\inf_{\Prob^m \in \mathcal{P}(\mu,d)^m}\frac{\mathrm{BUND}(p_m^*,\Prob^m)}{m}  \leq \mu - \frac{p_m^*}{m}(1 - f(\mu,d)/(m \epsilon_m^4)) \\
&= \mu - (1-\epsilon_m)^2\left(\mu - \frac{d}{2(1-\epsilon_m)}\right)(1 - f(\mu,d)/(m\epsilon_m^4)).
\end{align*}
Here we use the upper bound $\mathrm{OPT}(m,\Prob^m) \leq m\mu$ and $f(\mu,d)$ is as in Proposition \ref{prop:tail_sum}. Taking the limit of $m \rightarrow \infty$, we then find that 
$$
\lim_{m\rightarrow \infty} \left[\inf_{p_m} \sup_{\Prob^m \in \mathcal{P}(\mu,d)^m} \frac{\mathrm{OPT}(m,\Prob^m) - \mathrm{BUND}(p_m,\Prob^m)}{m} \right] \leq \frac{d}{2}.
$$

\noindent We next fix $m$ and consider a  case distinction in order to prove that $d/2$ is also a lower bound.

\textbf{Case 1:} $p_m < m(\mu-d/2)$. Note that  $\mathrm{BUND}(p_m,\Prob^m) \leq p_m < m(\mu-d/2)$ always holds. Also, we have for any $0 < \gamma < 1$, it holds that 
$$
\sup_{\Prob^m \in \mathcal{P}(\mu,d)^m} \frac{\mathrm{OPT}(m,\Prob^m)}{m} \geq (1-\gamma )\mu \cdot\left[1 - \frac{g(\mu,d)}{(\gamma \mu)^2 m}\right].
$$

Then we have 
\begin{align*}
\sup_{\Prob^m \in \mathcal{P}(\mu,d)^m} \frac{\mathrm{OPT}(m,\Prob^m) - \mathrm{BUND}(p_m,\Prob^m)}{m} &\geq \sup_{\Prob^m \in \mathcal{P}(\mu,d)^m} \frac{\mathrm{OPT}(m,\Prob^m) - m(\mu-d/2)}{m}\\
& = \sup_{\Prob^m \in \mathcal{P}(\mu,d)^m} \frac{\mathrm{OPT}(m,\Prob^m)}{m} - (\mu - d/2) \\
& \geq (1-\gamma )\mu \cdot\left[1 - \frac{g(\mu,d)}{(\gamma \mu)^2 m}\right] - (\mu  -d/2).
\end{align*}

\textbf{Case 2:} $p_m \geq m(\mu - d/2)$. In this case, we consider the joint distribution $\Prob^m_{\alpha} = \Pi_{i=1}^m \mathbb{P}_{\alpha}$ with $\mathbb{P}_{\alpha}$ the two-point distribution as in \eqref{eq:two_point_mean_mad} supported on $\{x(\alpha),y(\alpha)\}$ with $\alpha \rightarrow 1$ (this happens independent of $m$, which is fixed). Note that
\begin{align*}
\sup_{\Prob^m \in \mathcal{P}(\mu,d)^m} \frac{\mathrm{OPT}(m,\Prob^m) - \mathrm{BUND}(p_m,\Prob^m)}{m}  \geq \frac{\mathrm{OPT}(m,\Prob^m_{\alpha}) - \mathrm{BUND}(p_m,\Prob^m_{\alpha})}{m}.
\end{align*}
Because $m$ is fixed, it follows that $\Prob^m_{\alpha}(\sum_{i=1}^m X_i \geq m(\mu-d/2)) \rightarrow 0$ as $\alpha \rightarrow 1$, since all probability mass cumulates on the smallest support point $mx(\alpha) < m(\mu-d/2)$. This means that $\mathrm{BUND}(p_m,\Prob^m_{\alpha}) \rightarrow 0$ as $\alpha \rightarrow 1$. 

Furthermore, we can lower bound $\mathrm{OPT}(m,\Prob^m_{\alpha})$ by the better option of two bundling mechanisms. Firstly, the bundling mechanism that sets the bundling price $p_{m}$ just below $mx(\alpha)$ yields a revenue of $\mu - d/2$ as $\alpha \rightarrow 1$. Secondly, if we choose the bundling price $p_{m}$ just below the second support point of $Y = \sum_{i=1}^m X_i$, namely $(m-1)x(\alpha) + y(\alpha)$, then the revenue would be
\begin{align*}
[(m-1)x(\alpha) + y(\alpha)]\Prob^m_{\alpha}\left(Y \geq (m-1)x(\alpha) + y(\alpha)\right) & = [(m-1)x(\alpha) + y(\alpha)]\left(1 - \Prob^m_{\alpha}\left(Y = mx(\alpha)\right)\right) \\
 &  = [(m-1)x(\alpha) + y(\alpha)](1 - \alpha^m)  \\
 & = \mu\frac{1 - \alpha^m}{m} + \frac{d}{2}\left[\frac{1-\alpha^m}{(1-\alpha)m} - \frac{(1-\alpha^m)(m-1)}{m\alpha}\right],
\end{align*}
where the last equality follows by the definitions of $x(\alpha)$ and $y(\alpha)$. Since $\alpha \rightarrow 1$, it holds that $$
\lim_{\alpha \rightarrow 1} \frac{1 - \alpha^m}{m} = 0, \ \  \lim_{\alpha \rightarrow 1}  \frac{1-\alpha^m}{(1-\alpha)m} = 1, \ \ \text{ and } \ \ \lim_{\alpha \rightarrow 1} \frac{(1-\alpha^m)(m-1)}{m\alpha} = 0,
$$
meaning that the revenue will approach $d/2$. From this it follows that $\lim_{\alpha \rightarrow 1} \mathrm{OPT}(m,\Prob^m_{\alpha}) \geq \max\{\mu-d/2,d/2\}$, and, hence
$$
\lim_{\alpha \rightarrow 1}  \frac{\mathrm{OPT}(m,\Prob^m_{\alpha}) - \mathrm{BUND}(p_m,\Prob^m_{\alpha})}{m}  \geq \max\{\mu-d/2,d/2\}.
$$
Combining Cases 1 and 2 yields that, for any $0 < \gamma < 1$,
\begin{align*}
&\inf_{p_m} \sup_{\Prob^m \in \mathcal{P}(\mu,d)^m} \frac{\mathrm{OPT}(m,\Prob^m) - \mathrm{BUND}(p_m,\Prob^m)}{m} \\
&\geq \min\left\{(1-\gamma )\mu \cdot\left[1 - \frac{g(\mu,d)}{(\gamma \mu)^2 m}\right] - (\mu  -d/2), \max\{\mu-d/2,d/2\}\right\},
\end{align*}
and then
\begin{align*}
\lim_{m\rightarrow \infty} \left[\inf_{p_m} \sup_{\Prob^m \in \mathcal{P}(\mu,d)^m} \frac{\mathrm{OPT}(m,\Prob^m) - \mathrm{BUND}(p_m,\Prob^m)}{m}\right] \\
\geq \min\left\{(1-\gamma )\mu - (\mu  -d/2), \max\{\mu-d/2,d/2\}\right\}.
\end{align*}
Because this holds for all $0 < \gamma < 1$, by letting $\gamma \rightarrow 0$ it follows that
\begin{align*}
\lim_{m\rightarrow \infty} \left[\inf_{p_m} \sup_{\Prob^m \in \mathcal{P}(\mu,d)^m} \frac{\mathrm{OPT}(m,\Prob^m) - \mathrm{BUND}(p_m,\Prob^m)}{m}\right] 
\geq \min\left\{d/2, \max\{\mu-d/2,d/2\}\right\} = d/2.
\end{align*}
This completes the proof.
\end{proof}

\section{Insights from the single-good setting}\label{sec:single_item}
We now depart from the multi-good setting and consider the single-good case. We characterize the optimal prices under expected revenue, competitive ratio, and absolute regret, and then prove Proposition \ref{th:3p_m1}.

Optimal prices for expected revenue \citep{roos2019chebyshev} and competitive ratio \citep{kleer2024distribution} are given by
\begin{align}
    p^*_{rev} = \mu\left(1 - \sqrt{\frac{d}{2\mu}}\right) \quad \text{and} \quad p^*_{cr} = \begin{cases}
        \mu\left(1 - \frac{\sqrt{d(8\mu+d)}-d}{4\mu}\right),\quad& d \leq (2\sqrt{5}-4)\mu,\\
        \mu\left(\frac{2\mu-d}{2\mu+d}\right),\quad& d \geq (2\sqrt{5}-4)\mu.
    \end{cases}
\end{align}

The optimal price for absolute regret has not previously been characterized in the literature. We therefore derive it from first principles with the full derivation provided in Appendix \ref{ap:regret}. We remark that when the seller is indifferent among prices, we select the smallest one. This tie-breaking rule serves only to preserve continuity of the optimal price without affecting any result or interpretation.

\begin{proposition}\label{prop:regret_m1}
    Let $\mu > 0$ and $0 \leq d < 2\mu$. Then
    \begin{align}\label{eq:regret_price}
    p^*_{ar}=\arg\inf_p\sup_{\bP\in\cP(\mu,d)}\textup{OPT}(\bP)-p\bP(X\geq p)=\begin{cases}
         \mu\left(\frac{2\mu-\sqrt{2\mu d}}{2\mu-d}\right), \quad& d \in [0,(3-\sqrt{5})\mu], \\
        \mu\left(\frac{d}{2\mu-d}\right), \quad& d \in [(3-\sqrt{5})\mu,\mu], \\
        \mu\left(\frac{\mu}{2\mu-d}\right), \quad& d \in [\mu,2\mu).
    \end{cases}
\end{align}
\end{proposition}

We stated earlier that all our results are decreasing in MAD, and that therefore one can replace the exact MAD $d$ by an upper bound without loss of generality. However, this is not true for Proposition \ref{prop:regret_m1}, which we use to prove Proposition \ref{th:3p_m1}. We therefore derive this result in Appendix \ref{ap:regret_m1_ub} for an upper bound on the MAD as well and show that this does not affect Proposition \ref{th:3p_m1}.

Let us now consider prices at or above $\mu-d/2$. In this case, Nature can choose a two-point distribution that places all probability mass below the posted price, resulting in zero expected revenue for the seller. Meanwhile, the optimal expected revenue under full information remains positive, albeit small. This outcome renders the competitive ratio equal to zero, even though the outcome is favorable under absolute regret because the absolute revenue loss is small. This clarifies why different objectives can lead to markedly different optimal prices. Moreover, Proposition \eqref{prop:regret_m1} establishes that $p^*_{ar} \geq \mu-d/2$ when $d \geq (3-\sqrt{5})\mu$. Combined with the preceding observation, this implies that a seller who optimizes for absolute regret over these MAD values has simultaneously zero expected revenue and a competitive ratio of zero. This explains the appeal of prices that perform well across multiple criteria, such as the bundling price \eqref{def:price} in the many-goods setting.

We now establish the proof of Proposition \ref{th:3p_m1}. First, we show that $p^*_{rev} < p^*_{cr}$ when $d \leq (2\sqrt{5}-4)\mu$. Hence, we need to show that $$\mu\left(1-\sqrt{\frac{d}{2\mu}}\right) < \mu\left(1-\frac{\sqrt{8d\mu+d^2}-d}{4\mu}\right),$$ which is equivalent to showing $$4\mu\sqrt{\frac{d}{2\mu}} = \sqrt{8d\mu} > \sqrt{8d\mu+d^2}-d.$$ This is always true, as $d > 0$.
    
Second, we show that $p^*_{rev} < p^*_{cr}$ when $d > (2\sqrt{5}-4)\mu$. Hence, we need to show that $$\mu\left(1-\sqrt{\frac{d}{2\mu}}\right) < \mu\left(\frac{2\mu-d}{2\mu+d}\right),$$ which is equivalent to showing $$\sqrt{\frac{d}{2\mu}} > \frac{2d}{2\mu+d}.$$ Squaring both sides and simplifying leads to $8d\mu < (2\mu+d)^2$, which is equivalent to $(2\mu-d)^2>0$.

Third, we show that $p^*_{cr} < p^*_{ar}$ when $d \leq (2\sqrt{5}-4)\mu$. Hence, we need to show that $$\frac{\sqrt{8d\mu+d^2}-d}{4\mu} > \frac{\sqrt{2d\mu}-d}{2\mu-d}.$$ However, notice that $$\frac{\sqrt{8d\mu+d^2}-d}{4\mu} > \frac{\sqrt{8d\mu}-d}{4\mu} = \frac{2\sqrt{2d\mu}-d}{4\mu}.$$ Hence, we can show $$\frac{2\sqrt{2d\mu}-d}{4\mu} > \frac{\sqrt{2d\mu}-d}{2\mu-d}.$$ instead. This is equivalent to $2\sqrt{2d\mu}<2\mu+d$, which in turn is equivalent to $d < \mu + \frac{d^2}{4\mu}$, which is always true, as we assumed $d \leq (2\sqrt{5}-4)\mu < \mu$.

Fourth, we show that $p^*_{cr} < p^*_{ar}$ when $d \in \left((2\sqrt{5}-4)\mu,\mu\right]$. Hence, we need to show that $$\frac{2\mu-d}{2\mu+d} < \frac{2\mu-\sqrt{2d\mu}}{2\mu-d},$$ which is equivalent to showing $$\sqrt{2d\mu}<2\mu-\frac{(2\mu-d)^2}{2\mu+d} = \frac{d(6\mu-d)}{2\mu+d}.$$ Squaring and division by $d$ results in $2\mu < \frac{d(6\mu-d)^2}{(2\mu+d)^2}$, which is equivalent to $$h(d) := 2\mu(2\mu+d)^2-d(6\mu-d)^2<0.$$ Since $h(d)$ is a third degree polynomial, it has at most three real solutions. One can verify that these solutions are $d_1 = (6-4\sqrt{2})\mu$, $d_2 = 2\mu$, and $d_3 = (6+2\sqrt{2})\mu$. Clearly, $d_1 <(2\sqrt{5}-4)\mu<\mu<d_2 < d_3$ and hence, it suffices to check whether $h(d) < 0$ for some $d \in (d_1,d_2)$. Choosing $d = \mu$ results in $h(d) = -7$, which proves our claim. Finally, note that when $d > (3-\sqrt{5})\mu$, then $p^*_{ar} > \tau_1 > p^*_{cr}$.

\section{Conclusion and future directions}\label{sec:conclusion}
Taken together, our findings indicate that heavy-tailed uncertainty with infinite variance sharply alters the revenue guarantees achievable by any mechanism, eliminating asymptotic full mean extraction. Yet even in this more demanding distributionally robust environment, deterministic bundling remains optimal in the many-goods limit.

We further analyze deterministic bundling through the lenses of competitive ratio and absolute regret. Our results show that these criteria admit a single, robust bundle price when restricting attention to the class of deterministic bundling when $m \xrightarrow{} \infty$. What remains unknown, however, is the speed of this convergence. Quantifying these convergence rates would provide valuable insights into how quickly high-dimensional aggregation smooths out the difference between these three objectives. It also remains unknown whether deterministic bundling is optimal among \textit{all} truthful mechanisms. Proving this, however, requires a different approach than the one used for the expected revenue objective, since the min-max result for competitive ratio (and the max-min result for absolute regret) are trivial. Any upper bound would therefore be uninformative.

Overall, deterministic bundling remains a powerful tool under heavy-tailed ambiguity, offering strong revenue guarantees and pricing consistency among objectives. Simultaneously, the open questions identified in this work point to rich directions for future research on distributional robustness and high-dimensional mechanism design.

\begin{small}
\bibliography{bibbook2}
\end{small}

\newpage
\appendix
\section{Minimax analysis for $\Ptwomud$ when $\mu \leq d < 2\mu$}
\label{app:two-point_minimax}
\begin{proposition}
For $\mu \leq d < 2\mu$ and $X_1,\dots,X_m$ i.i.d., it holds that
\begin{align}
\lim_{m \rightarrow \infty} \inf_{\Prob^m \in \mathcal{P}_2(\mu,d)^m}\sup_{D \in \mathcal{D}_m} \ \, \frac{ \mathrm{REV}(D,\Prob^m)}{m}   & \geq \lim_{m \rightarrow \infty} \inf_{\Prob^m \in \mathcal{P}_2(\mu,d)^m}\sup_{p_m} \ \, \frac{ p_m \cdot \Prob^m(\sum_{i=1}^{m} X_i \geq p_m)}{m} \geq  \mu - \frac{d}{2} + \xi
\label{eq:asymptotic_p_main}
\end{align}
with $\xi = \xi(\mu,d) > 0$ some small number.
\label{cor:asymptotic_minmax}
\end{proposition}

\begin{proof}
As we are considering two-point distributions in an i.i.d. setting, the choice of $\Prob^m$ can be characterized by the choice of $\alpha = \alpha(m)$ in \eqref{eq:two_point_mean_mad}. Note that if the infimum is not attained we take a sequence $(\alpha(m))_{m \in \N}$ close enough to attaining the infimum, i.e., for which $\sup_{p} \ \, p \cdot \Prob^m_{\alpha(m)}(\sum_{i=1}^{m}X_i \geq p) \leq \delta +  \inf_{\Prob^m \in \mathcal{P}(\mu,d)}\sup_{p} \ \, p \cdot \Prob^m(\sum_{i=1}^{m}X_i \geq p)$ for an arbitrary $\delta > 0$. We write $\Prob^m = \Prob^m_{\alpha(m)}$ to emphasize this, and focus on  bounding
\begin{align}
\lim_{m \rightarrow \infty} \sup_{p_m} \ \, \frac{p_m \cdot \Prob^m_{\alpha(m)}(\sum_{i=1}^{m}X_i \geq p_m)}{m}.
\label{eq:inner_sup}
\end{align}
We analyse different regimes of growth of the sequence  $(\alpha(m))_{m \in \N}$, that together capture all possibilities.  To make notation a bit more convenient, we will write $\alpha(m) = 1 - 1/q(m)$ and make our case distinction in terms of $q(m)$.

We will show that for any growth regime of $q(m)$ chosen by nature, the seller can always guarantee a revenue of either $d/2$ or $\mu - d/2 + \xi$ for some small $\xi > 0$, in case $\mu < d < 2\mu$.  This yields the lower bound in \eqref{eq:mu-d/2} of $\mu - d/2 + \xi$, as this quantity is smaller than $d/2$. There will be three cases, that together capture all possibilities. We emphasize that at this point we assume $\mu < d < 2\mu$.
\begin{itemize}
\item Case 1: $\lim_{m \rightarrow \infty} q(m)/m = \infty$. \textit{We show the seller can guarantee a revenue of at least $d/2$.}
\item Case 2: $\lim_{m \rightarrow \infty} q(m)/m = 0$. \textit{We show the seller can guarantee a revenue of at least $d/2$.}
\item Case 3: $\lim_{m \rightarrow \infty} q(m)/m = \lambda$ for constant $\lambda > 0$.
\begin{itemize}
\item \textit{If $\lambda$ is sufficiently small, we show  the seller can guarantee a revenue of at least $d/2$.}
\item \textit{If $\lambda$ is not small, we show  the seller can guarantee a revenue of at least $\mu - d/2 + \xi$.}
\end{itemize}
\end{itemize} \medskip
The "bottleneck" in the analysis, that leads to only being able to guarantee $\mu - d/2 + \xi$ for the seller is caused by the second subcase of Case 3.

\textit{Case 1: $\lim_{m \rightarrow \infty} q(m)/m = \infty$.} We use a similar argument as that for the upper bound of $\max\{\mu-d/2,d/2\}$. Setting the price $p_m$ just below the second support point 
$(m-1)x(\alpha) + y(\alpha)$ gives a revenue of 
$$
\approx (1-\alpha^m)\left(\mu - \frac{d}{2}\right) + \frac{(1-\alpha^m)}{m(1-\alpha)}\frac{d}{2}.
$$
For any choice of $q(m)$ for which $q(m)/m \rightarrow \infty$ as $m \rightarrow \infty$, it can be shown that 
$$
(1 - \alpha^m) \rightarrow 0 \ \text{ and } \ (1 - \alpha^m)/(m(1-\alpha)) \rightarrow 1.
$$ 
This shows that the seller can guarantee a revenue of $d/2$ as $m \rightarrow \infty$ in this case. Since $d > \mu$, we can find $\xi_0(\mu,d) > 0$ such that $d/2 \geq \mu - d/2 + \xi_0$. \medskip

\textit{Case 2:  $\lim_{m \rightarrow \infty} q(m)/m = 0$.} Note that for $X(\alpha)$ as in \eqref{eq:two_point_mean_mad}, we have
$$
\text{Var}(X(\alpha)) = \alpha\left(\mu - \frac{d}{2\alpha} - \mu\right)^2 + (1-\alpha)\left(\mu - \frac{d}{2(1-\alpha)} - \mu\right)^2 = \frac{d^2}{4\alpha} + \frac{d^2}{4(1-\alpha)}  \approx \frac{d^2}{4}\left(1 + q(m)\right).
$$
The Chebyshev bound in Lemma \ref{lem:chebyshev} then tells us that
$$
\Prob^m_{\alpha(m)}\left(\sum_{i=1}^m X_i \leq (1-\gamma)m\mu\right) \leq  \frac{\text{Var}(X(\alpha))}{(\gamma\mu)^2m} = \frac{d^2(1+q(m))}{(2\gamma\mu)^2m} \rightarrow 0,
$$
as $m \rightarrow \infty$ by the assumption $\lim_{m \rightarrow \infty} q(m)/m = 0$. This means that if for $\gamma > 0$ the seller sets a price of $p_m = (1-\gamma)m\mu$, she can guarantee a revenue of roughly $p_m$ (as we sell with probability approaching $1$), i.e., 
$$
\lim_{m\rightarrow \infty} \sup_{p} \ \, \frac{p \cdot \Prob^m_{\alpha(m)}(\sum_{i=1}^{m}X_i \geq p)}{m} \geq (1-\gamma)\mu
$$
for any fixed $\gamma > 0$. If we choose $\gamma$ close enough to zero, we have $(1-\gamma)\mu \geq d/2 \geq \mu-d/2 + \xi_0$. \medskip

\textit{Case 3: $\lim_{m \rightarrow \infty} q(m)/m = \lambda$ for a fixed constant $\lambda > 0$.} 
First, we argue that if $\lambda$ is sufficiently small, we can still use Chebyshev's bound in order to argue that we can obtain a revenue somewhat close to $\mu$. Let $0 < \gamma < 1$, such that
\begin{align}
0.99(1- \gamma)\mu \geq \frac{d}{2},
\label{eq:gamma}
\end{align}
which is possible as $d < 2\mu$, and choose $\tau_0 > 0$ such that 
\begin{align}
1 - \left(\frac{d}{2\gamma\mu} \right)^{2}\tau_0 \geq 0.99.
\label{eq:tau}
\end{align}
Then for any $\lambda \in (0,\tau_0]$, Chebyshev's bound tells us that 
\begin{align*}
\lim_{m \rightarrow \infty} \frac{(1-\gamma)m\mu \cdot \displaystyle \Prob^m_{\alpha(m)}\left(\sum_{i=1}^m X_i \geq (1-\gamma)m\mu\right)}{m} & \geq  \lim_{m \rightarrow \infty} (1-\gamma)\mu \left(1 - \frac{d^2(1+q(m))}{(2\gamma\mu)^2m} \right) \\
 & = (1-\gamma)\mu \left(1 - \frac{d^2}{(2\gamma\mu)^2}\lambda \right) \\
  & \geq  (1-\gamma)\mu \left(1 - \frac{d^2}{(2\gamma\mu)^2}\tau_0 \right)\\
  & \geq 0.99(1-\gamma)\mu \geq \frac{d}{2} > \mu - d/2 + \xi_0,
\end{align*}
where the final two inequalities use the choice of $\gamma$ in \eqref{eq:gamma} and $\tau_0$ in \eqref{eq:tau}, respectively. This means that for $\lambda \in (0,\tau_0]$, the seller can guarantee a revenue of at least $d/2 \geq \mu - \frac{d}{2} + \xi_0$.

For $\lambda \in [\tau_0,\infty)$, we will again use a similar revenue analysis as in Case 1. 
We will argue that if we set a price just below $(m-1)x(\alpha) + y(\alpha)$, then the seller can guarantee a revenue strictly larger than $\mu - d/2$ (independent of $\lambda$). We have
\begin{align}
\frac{[(m-1)x(\alpha) + y(\alpha)] \cdot \Prob^m_{\alpha(m)}(\sum_{i=1}^{m}X_i \geq (m-1)x(\alpha) + y(\alpha))}{m} 
& \approx (1-\alpha^m)\left(\mu - \frac{d}{2}\right) + \frac{(1-\alpha^m)}{m(1-\alpha)}\frac{d}{2}.
\label{eq:calc_k2_case3}
\end{align}
Note that 
$$
1 - \alpha^m = 1 - \left(1 - \frac{1}{q(m)}\right)^m \rightarrow 1 - e^{-1/\lambda},
$$
as $m \rightarrow \infty$, using $\lim_{n \rightarrow \infty} (1 + x/n)^n = e^x$, and similarly 
$$
\frac{(1-\alpha^m)}{m(1-\alpha)} \rightarrow \lambda \left(1 - e^{-1/\lambda}\right),
$$
as $m \rightarrow \infty$, so that 
\begin{align*}
\lim_{m \rightarrow \infty} \frac{[(m-1)x(\alpha) + y(\alpha)] \cdot \Prob^m_{\alpha(m)}(\sum_{i=1}^{m}X_i \geq (m-1)x(\alpha) + y(\alpha))}{m} \\
= \left(1 - e^{-1/\lambda}\right)\left(\mu + (\lambda-1)\frac{d}{2}\right) =: g_{\mu,d}(\lambda).
\end{align*}
The function $g_{\mu,d}(\lambda)$ has the property that
$$
g_{\mu,d}(\lambda) \rightarrow 
\left\{\begin{array}{ll}
\mu - d/2 & \text{ if } \lambda \rightarrow 0\\
d/2 & \text{ if } \lambda \rightarrow \infty
\end{array}\right.,
$$
and it remains bounded away from $\mu - d/2$ on $[\tau_0,\infty)$. To be precise, we can find a $\xi_1(\mu,d,\tau_0)$ such that
$$
\left(1 - e^{-1/\lambda}\right)\left(\mu + (\lambda-1)\frac{d}{2}\right) \geq  \mu - \frac{d}{2} + \xi_1(\mu,d,\tau_0)
$$
for all $\lambda \in [\tau_0,\infty)$.
We can then define $\xi = \min\{\xi_0,\xi_1\}$.
\end{proof}

\section{Proof of Proposition \ref{prop:regret_m1}}\label{ap:regret}
The proof is structured as follows: First, we specify the probability distributions used in the proof. Second, for fixed $p$, we solve the inner maximization problem over all feasible probability distributions. Finally, we solve the outer minimization problem over all $p$.

\subsection{Probability distributions (primal solutions)}\label{ap:pdps}
We introduce (and restate) the probability distributions that play a role in Proposition \ref{prop:regret_m1}. First, we parametrize the class of two-point distributions with fixed mean-MAD.
\begin{definition}\label{def:2pd}
    Consider a two-point distribution with given mean $\mu$ and mean absolute deviation $d$. Denote $\tau_1 = \mu - \frac{d}{2}$ and $\tau_2 = \mu+\frac{d\mu}{2\mu-d}$. Then
\begin{align}
\bP^x_{2}(x) = \left\{ \begin{array}{ll}
x & \text{ w.p. } v_1 = \frac{d}{2(\mu-x)} \\
y(x) = \mu+\frac{d(\mu-x)}{2(\mu-x)-d} & \text{ w.p. } v_2 = 1-\frac{d}{2(\mu-x)}
\end{array}\right., \textit{ for } x \in [0,\tau_1),
\end{align}
and
\begin{align}
\bP^y_{2}(y) = \left\{ \begin{array}{ll}
x(y) = \mu-\frac{d(y-\mu)}{2(y-\mu)-d} & \text{ w.p. } v_1 = 1-\frac{d}{2(y-\mu)} \\
y & \text{ w.p. } v_2 = \frac{d}{2(y-\mu)} 
\end{array}\right., \textit{ for } y \in [\tau_2,\infty).
\end{align}
\end{definition}

Notice that distribution \eqref{eq:two_point_mean_mad}, $\bP^x_{2}(x)$, and $\bP^y_{2}(y)$ are the same distributions with different parametrizations. This distinction is convenient in the proof of Proposition \ref{prop:regret_m1}. Next, we parametrize all possible three-point distributions in $\cP(\mu,d)$.
\begin{definition}\label{def:3pd}
Consider a three-point distribution with given mean $\mu$ and mean absolute deviation $d$. Denote $\bar{\tau}_1(z) = \mu - \frac{d(z - \mu)}{2(z - \mu) - d}$ and $\bar{\tau}_2(x) = \mu + \frac{d(\mu - x)}{2(\mu - x) - d}$. Then
\begin{align}
\bP^-_{3}(x,y,z) = \left\{ \begin{array}{ll}
x & \text{ w.p. } w^-_1 = \frac{d(z-y)-2(\mu-y)(z-\mu)}{2(z-\mu)(y-x)} \\
y & \text{ w.p. } w^-_2 = \frac{2(\mu-x)(z-\mu)-d(z-x)}{2(z-\mu)(y-x)} \\
z & \text{ w.p. } w^-_3 = \frac{d}{2(z-\mu)}
\end{array}
\right.
\end{align}
with $x \in [0,\bar{\tau}_1(z))$, $y\in[\bar{\tau}_1(z),\mu]$, and $z \in [\tau_2,\infty)$, and
\begin{align}
\bP^+_{3}(x,y,z) = \left\{ \begin{array}{ll}
x & \text{ w.p. } w^+_1 = \frac{d}{2(\mu - x)} \\
y & \text{ w.p. } w^+_2 = \frac{d(z - x)-2(z-\mu)(\mu - x)}{2(\mu - x)(y-z)} \\
z & \text{ w.p. } w^+_3 = \frac{2(y-\mu)(\mu - x)-d(y-x)}{2(\mu - x)(y-z)}
\end{array}\right.
\end{align}
with $x \in [0,\tau_1)$, $y\in[\mu,\bar{\tau}_2(x))$, and $z \in [\bar{\tau}_2(x),\infty)$.
\end{definition}

\subsection{Inner maximization problem}
\begin{lemma}\label{lem:sup_ar}
    Let $p > 0$ be fixed, and denote $\tau_1 = \mu - \frac{d}{2}$, $\tau_2 = \mu+\frac{d\mu}{2\mu-d}$. Then
    \begin{align}\label{eq:sup_ar2}
    \sup_{\bP\in\cP(\mu,d)}\textup{OPT}(\bP)-p\bP(X\geq p) = \begin{cases}
        \max\left\{\frac{dp}{2(\mu-p)},\mu-p+\frac{dp}{2\mu}\right\}, \quad& p \in (0,\tau_1], \\
        \max\left\{\mu-\frac{d}{2},\mu-p+\frac{dp}{2\mu}\right\}, \quad& p \in [\tau_1,\mu], \\
        \max\left\{\left(\mu-\frac{d}{2}\right)\frac{p}{\mu},\mu-p+\frac{dp}{2\mu}\right\}, \quad& p \in [\mu,\tau_2], \\
        \mu, \quad& p \in [\tau_2,\infty).
    \end{cases}
\end{align}
\end{lemma}
\begin{proof}
Consider
\begin{align}\label{eq:swapped2}
    \sup_{\bP \in \cP(\mu,d)} \textup{OPT}(\bP) - p\bP(X\geq p) &= \sup_{\bP \in \cP(\mu,d)} \sup_t t\bP(X\geq t) - p\bP(X> p) \nonumber \\
    &= \sup_t\sup_{\bP \in \cP(\mu,d)} t\bP(X\geq t) - p\bP(X> p).
\end{align}
The first equality follows from $\sup_{\bP \in \cP(\mu,d)} \textup{OPT}(\bP) - p\bP(X > p)$ being continuous in $p$, which we verify at the end, while the second equality holds by changing the maximization order. For the remainder of this proof, we consider $p$ and $t$ fixed and derive a primal-dual formulation of \eqref{eq:sup_ar2}, establishing matching bounds. 
    
Define $\mathcal{M}$ as the set of all probability measures defined on the measurable space $([0,\infty),\mathcal{B})$ with $\mathcal{B}$ the Borel sigma-algebra on $\R_+$. Then, for fixed $p$ and $t$,
\begin{equation}\label{primalsup2}
\begin{aligned}
&\! \sup_{\mathbb{P}\in \mathcal{M}} &  &\int_x t\1_{\{x\geq t\}}-p\1_{\{x> p\}}{\rm d} \mathbb{P}(x),\\
&\text{s.t.} &      &  \int_x {\rm d}\mathbb{P}(x)=1, \ \int_x x{\rm d}\mathbb{P}(x)=\mu,\ \int_x |x-\mu|{\rm d}\mathbb{P}(x)=d.
\end{aligned}
\end{equation}
Consider the corresponding dual, 
\begin{equation}\label{dualsup2}
\begin{aligned}
&\inf_{\lambda_0,\lambda_1,\lambda_2 \in \R}\quad
 \lambda_0 + \lambda_1 \mu + \lambda_2 d,\\
&\text{s.t.}\quad
 G(x) := t\1_{\{x\geq t\}}-p\1_{\{x> p\}}
 \leq \lambda_0 + \lambda_1 x + \lambda_2 |x-\mu| =: F(x),
\quad \forall x\in[0,\infty).
\end{aligned}
\end{equation}
The goal is to construct matching primal and dual pairs for each $(p,t) \in (0,\infty)\times(0,\infty)$. To this end, we partition $(0,\infty) \times (0,\infty)$ into a finite set of scenarios, each of which is analyzed independently. Within Scenario $i$, we denote the primal and dual objectives by $r^p_i(\bP,p,t)$ and $r^d_i(\lambda,p,t)$, respectively. We refer to the distributions in Appendix \ref{ap:pdps} when specifying the primal solutions. In case the primal and dual solutions match, we write $r_i(p,t) := r^d_i(\lambda,p,t) = r^p_i(\bP,p,t)$. We then maximize over $t$ and write $r^*_i(p) = \sup_t r_i(p,t)$. Finally, we obtain \eqref{eq:swapped2} by solving $\max_i r^*_i(p)$.

\begin{figure}[h!]
  \centering
  \begin{subfigure}[b]{0.3\linewidth}
    \centering
    \includegraphics[width=\linewidth]{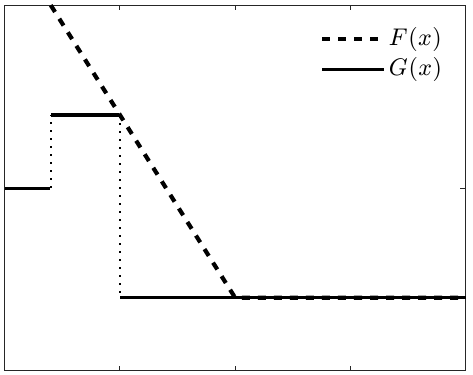}
    \caption{$1a$: $p \in (0,\tau_1),t \in (0,p]$.}
  \end{subfigure}\hfill
  \begin{subfigure}[b]{0.3\linewidth}
    \centering
    \includegraphics[width=\linewidth]{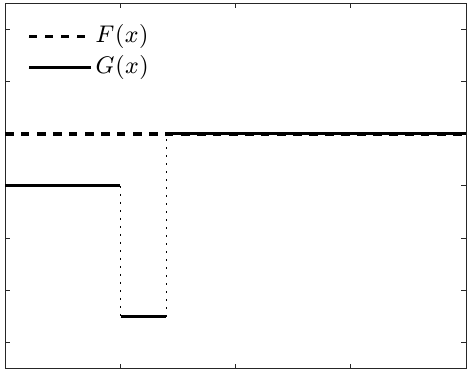}
    \caption{1b: $p \in (0,\tau_1),t \in (p,\tau_1]$.}
  \end{subfigure}\hfill
  \begin{subfigure}[b]{0.3\linewidth}
    \centering
    \includegraphics[width=\linewidth]{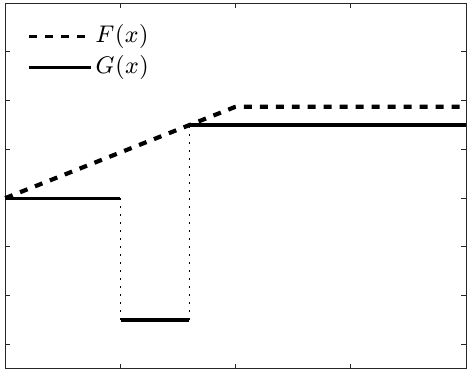}
    \caption{1c: $p \in (0,\tau_1),t \in (\tau_1,\mu]$.}
  \end{subfigure}

  \vspace{1ex}

  \begin{subfigure}[b]{0.3\linewidth}
    \centering
    \includegraphics[width=\linewidth]{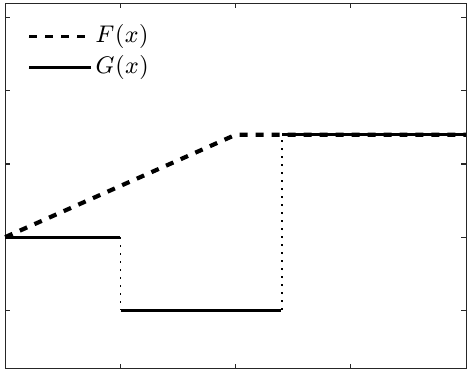}
    \caption{1d: $p \in (0,\tau_1),t \in (\mu,\tau_2]$.}
  \end{subfigure}\hfill
  \begin{subfigure}[b]{0.3\linewidth}
    \centering
    \includegraphics[width=\linewidth]{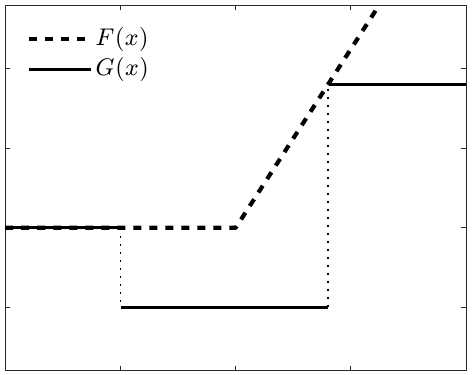}
    \caption{1e: $p \in (0,\tau_1),t \in (\tau_2,\bar{\tau}_2(p)]$.}
  \end{subfigure}\hfill
  \begin{subfigure}[b]{0.3\linewidth}
    \centering
    \includegraphics[width=\linewidth]{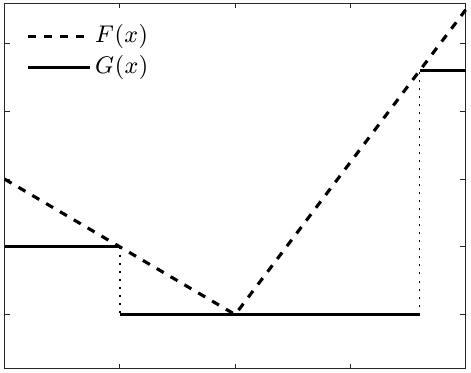}
    \caption{1f: $p \in (0,\tau_1), t \in (\bar{\tau}_2(p),\infty)$.}
  \end{subfigure}

  \vspace{1ex}

  \begin{subfigure}[b]{0.3\linewidth}
    \centering
    \includegraphics[width=\linewidth]{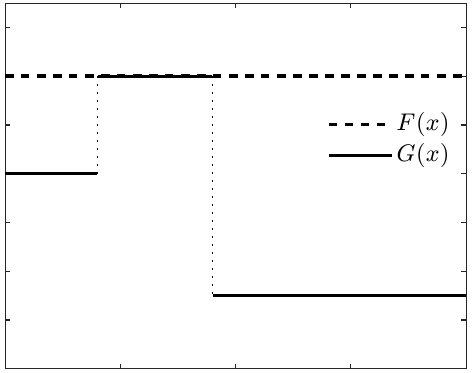}
    \caption{2a: $p \in [\tau_1,\mu),t \in (0,\tau_1]$.}
  \end{subfigure}\hfill
  \begin{subfigure}[b]{0.3\linewidth}
    \centering
    \includegraphics[width=\linewidth]{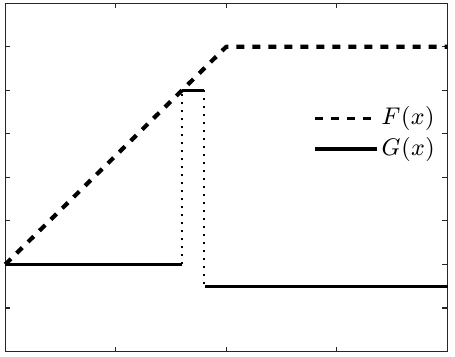}
    \caption{2b: $p \in [\tau_1,\mu),t \in (\tau_1,p]$.}
  \end{subfigure}\hfill
  \begin{subfigure}[b]{0.3\linewidth}
    \centering
    \includegraphics[width=\linewidth]{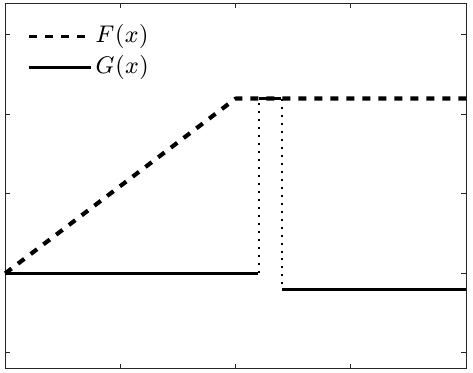}
    \caption{3c: $p \in [\mu,\tau_2),t \in (\mu,p]$.}
  \end{subfigure}

  \caption{Illustration of all nine distinct scenarios with $\mu = 1$ and $d=0.5$.}
  \label{fig:case1}
\end{figure}

\textbf{Case 1:} $p \in (0,\tau_1)$. First, consider $t \in (0,p]$, which we call Scenario 1a. Consider $F(p) = \lambda_0 + \lambda_1 p + \lambda_2 (\mu-p) = t$, $F(\mu)=\lambda_0 + \lambda_1 \mu=t-p$, and $F'_{+}(\mu)=\lambda_1+\lambda_2=0$.

\noindent This results in
\begin{align*}
    \lambda_0 = \frac{p(2t+\mu)-2\mu t - 2p^2}{2(p-\mu)}, \ \ \lambda_1 = -\frac{p}{2(\mu-p)}, \ \ \lambda_2 = \frac{p}{2(\mu-p)}
\end{align*}
with corresponding dual value $r^d_{1a}(\lambda,p,t) = t - p(1-\frac{d}{2(\mu-p)})$. Now consider $\bP = \bP^x_{2}(p)$ and notice that this results in a matching primal objective $r^p_{1a}(\bP,p,t)= t - p(1-v_1) = r^d_{1a}(\lambda,p,t)$. Clearly, $r_{1a}(p,t)$ is increasing in $t$. Hence, we only need to consider $t = p$ with $$r_{1a}^*(p) = \frac{dp}{2(\mu-p)}.$$

Second, consider $t \in (p,\tau_1]$, which we call Scenario 1b. Consider $F(t)=\lambda_0+\lambda_1 t + \lambda_2 (\mu-t) = t-p$, $F(\mu)=\lambda_0+\lambda_1 \mu=t-p$, and $F'_-(\mu)=\lambda_1-\lambda_2=0$. This results in 
\begin{align*}
    \lambda_0 = t-p, \ \ \lambda_1 = 0, \ \ \lambda_2 = 0
\end{align*}
with corresponding dual value $r^d_{1b}(\lambda,p,t) = t-p$. Now, when $t\in (p,\tau_1)$, consider $\bP = \bP^x_{2}(\tau_1-\epsilon)$ with $\epsilon \in (0,\tau_1-t)$. Then $\bP(X\geq t)=\bP(X>p) = 1$, and hence $r^p_{1b}(\bP,p,t) = t-p = r^d_{1b}(\lambda,p,t)$. Moreover, choosing $t=\tau_1-\epsilon$ and letting $\epsilon\downarrow{}0$ yields $t-p \uparrow \tau_1-p$, and thus $\sup_{t\in(p,\tau_1]} r_{1b}(p,t)=\tau_1-p$. Therefore, we obtain 
$$r_{1b}^*(p) = \mu - \frac{d}{2}-p.$$

Third, consider $t \in (\tau_1,\mu]$, which we call Scenario 1c. Consider $F(0) = \lambda_0+\lambda_2 \mu = 0$, $F(t) = \lambda_0+\lambda_1 t + \lambda_2(\mu-t)=t-p$, and $F'_+(\mu) = \lambda_1 + \lambda_2=0$. This results in
\begin{align*}
    \lambda_0 = \frac{\mu(t-p)}{2t}, \ \ \lambda_1 = \frac{t-p}{2t}, \ \ \lambda_2 = -\frac{t-p}{2t}
\end{align*}
with corresponding dual value $r^d_{1c}(\lambda,p,t) = \frac{(t-p)(2\mu-d)}{2t}$. Now consider $\bP = \bP^-_{3}(0,t,z)$ and let $z \xrightarrow{} \infty$ and notice that $\lim_{z\xrightarrow{}\infty}r^p_{1c}(\bP,p,t)= \lim_{z\xrightarrow{}\infty}(t-p)(w_2^-+w_3^-) = \frac{(t-p)(2\mu-d)}{2t} = r^d_{1c}(\lambda,p,t)$. Henceforth, we write $\bP=\bP_3^-(0,t,\infty)$ (or $\bP=\bP_3^+(0,t,\infty)$) as shorthand for using such limiting sequences. By taking the derivative of $r_{1c}(p,t)$ with respect to $t$ we find $\frac{\partial r_{1c}(p,t)}{\partial t} = \frac{p(2\mu-d)}{2t^2} \geq 0$. Hence, we only need to consider $t = \mu$ with objective value $$r^*_{1c}(p) = (\mu-p)\frac{2\mu-d}{2\mu}.$$

Fourth, consider $t \in (\mu,\tau_2]$, which we call Scenario 1d. Consider $F(0) = \lambda_0 + \lambda_2\mu = 0$, $F(t)=\lambda_0 + \lambda_1 t+\lambda_2 (t-\mu)=t-p$, and $F'_{+}(\mu)=\lambda_1+\lambda_2=0$. This results in
\begin{align*}
    \lambda_0 = \frac{t-p}{2}, \ \ \lambda_1 = \frac{t-p}{2\mu}, \ \ \lambda_2 = -\frac{t-p}{2\mu}
\end{align*}
with corresponding dual value $r^d_{1d}(\lambda,p,t) = \frac{(2\mu-d)(t-p)}{2\mu}$. Now consider $\bP = \bP^x_{2}(0)$ and notice that this results in a matching primal objective $r^p_{1d}(\bP,p,t)= (t-p)v_2 = r^d_{1d}(\lambda,p,t)$. By taking the derivative of $r_{1d}(p,t)$ with respect to $t$ we find $\frac{\partial r_{1d}(p,t)}{\partial t} = 1 - \frac{d}{2\mu} \geq 0$. Hence, we only need to consider $t = \tau_2$ with objective value $$r^*_{1d}(p) = \mu - p + \frac{dp}{2\mu}.$$

Fifth, consider $t \in (\tau_2,\bar{\tau}_2(p)]$ with $\bar{\tau}_2(p) = \mu+\frac{d(\mu-p)}{2(\mu-p)-d}$, which we call Scenario 1e. Consider $F(\mu)=\lambda_0+\lambda_1 \mu =0$, $F(t)=\lambda_0+\lambda_1 t + \lambda_2 (t-\mu)=t-p$, and $F_{-}'(\mu)=\lambda_1-\lambda_2=0$. This results in
\begin{align*}
    \lambda_0 = -\frac{\mu(t-p)}{2(t-\mu)}, \ \ \lambda_1 = \frac{t-p}{2(t-\mu)}, \ \ \lambda_2 = \frac{t-p}{2(t-\mu)}
\end{align*}
with corresponding dual value $r^d_{1e}(\lambda,p,t) = \frac{d(t-p)}{2(t-\mu)}$. Now consider $\bP = \bP^y_{2}(t)$ and notice that this results in a matching primal objective $r^p_{1e}(\bP,p,t)= (t-p)v_2 = r^d_{1e}(\lambda,p,t)$. By taking the derivative of $r_{1e}(p,t)$ with respect to $t$ we find $\frac{\partial r_{1e}(p,t)}{\partial t} = -\frac{d(\mu-p)}{2(t-\mu)^2}\leq 0$. Therefore, we only need to consider $t = \tau_2^+$ with objective value $$r^*_{1e}(p) = \mu - p + \frac{dp}{2\mu}.$$

Sixth, consider $t \in (\bar{\tau}_2(p),\infty)$, which we call Scenario 1f. Consider $F(p) = \lambda_0 + \lambda_1 p + \lambda_2 (\mu-p) = 0$, $F(\mu)=\lambda_0 + \lambda_1 \mu =-p$, and $F(t)=\lambda_0 + \lambda_1 t + \lambda_2 (t-\mu)=t-p$. This results in
\begin{align*}
    \lambda_0 = \frac{2p^2(t-\mu)+\mu^2p - \mu^2 t}{2(p-\mu)(\mu-t)}, \ \ \lambda_1 = \frac{p(\mu-2t)+\mu t}{2(p-\mu)(\mu-t)}, \ \ \lambda_2 = \frac{\mu(t-p)}{2(p-\mu)(\mu-t)}
\end{align*}
with corresponding dual value $r^d_{1f}(\lambda,p,t) = \frac{dt}{2 \left(t - {\mu}\right)} - p \left(1 - \frac{d}{2 \left({\mu} - p\right)}\right)$. Now consider $\bP = \bP^+_{3}(p,\mu,t)$ and notice that this results in a matching primal objective $r^p_{1f}(\bP,p,t)= tw_3^+-p(w_2^++w_3^+) = r^d_{1f}(\lambda,p,t)$. By taking the derivative of $r_{1f}(p,t)$ with respect to $t$ we obtain $\frac{\partial r_{1f}(p,t)}{\partial t} = -\frac{d\mu}{2(t-\mu)^2} \leq 0$. Hence, we only need to consider $t = \bar{\tau}_2(p)$ with objective value $$r^*_{1f}(p) = \mu - p.$$

Finally, maximizing over all scenarios gives 
\begin{align*}
\sup_t\sup_{\bP \in \cP(\mu,d)} t\bP(X\geq t) - p\bP(X> p) &= \max\{r^*_{1a}(p),r^*_{1b}(p),r^*_{1c}(p),r^*_{1d}(p),r^*_{1e}(p),r^*_{1f}(p)\}\\
&=\max\{r^*_{1a}(p),r^*_{1d}(p)\}\\
&= \max\left\{\frac{dp}{2(\mu-p)},\mu-p+\frac{dp}{2\mu}\right\},
\end{align*}
since $r^*_{1b}(p) = \mu-\frac{d}{2}-p \leq \mu-p + \frac{dp}{2\mu} = r^*_{1d}(p)$, $r^*_{1c}(p) = (\mu-p)\frac{2\mu-d}{2\mu} \leq \mu-p + \frac{dp}{2\mu} = r^*_{1d}(p)$, $r^*_{1e}(p) = r^*_{1d}(p)$, and $r^*_{1f}(p) = r^*_{1d}(p)$.
\medskip

\textbf{Case 2: $p \in [\tau_1,\mu)$.} First, consider $t \in (0,\tau_1]$, which we call Scenario 2a. Consider $F(0) = \lambda_0 + \lambda_2\mu = t$, $F'_{-}(\mu)=\lambda_1-\lambda_2=0$, and  $F'_{+}(\mu)=\lambda_1+\lambda_2=0$. This results in
\begin{align*}
    \lambda_0 = t, \ \ \lambda_1 = 0, \ \ \lambda_2 = 0
\end{align*}
with corresponding dual value $r^d_{2a}(\lambda,p,t) = t$. Now, when $t \in (0,\tau_1)$, consider $\bP = \bP^x_{2}(\tau_1-\epsilon)$ with $\epsilon \in (0,\tau_1-t)$. Then, letting $\epsilon \downarrow{} 0$ results in $\lim_{\epsilon\downarrow{}0 }t\bP(X\geq t) -p\bP(X>p) = t-p\lim_{\epsilon \downarrow{} 0}v_2=t$, and hence $\lim_{\epsilon \downarrow{}0}r^p_{2a}(\bP,p,t) = t = r^d_{2a}(\lambda,p,t)$. Moreover, choosing $t=\tau_1-\epsilon$ and letting $\epsilon\downarrow{}0$ yields $t \uparrow \tau_1$, and thus $\sup_{t\in(0,\tau_1]} r_{2a}(p,t)=\tau_1$. Hence, $$r_{2a}^*(p) = \mu - \frac{d}{2}.$$

Second, consider $t \in (\tau_1,p]$, which we call Scenario 2b. Consider $F(0) = \lambda_0 + \lambda_2 \mu = 0$, $F(t)=\lambda_0 + \lambda_1 t + \lambda_2(\mu-t)=t$, and $F'_+(\mu)=\lambda_1 +\lambda_2=0$. This results in
\begin{align*}
    \lambda_0 = \frac{\mu}{2}, \ \ \lambda_1 = \frac{1}{2}, \ \ \lambda_2 = -\frac{1}{2}
\end{align*}
with corresponding dual value $r^d_{2b}(\lambda,p,t) = \mu-\frac{d}{2}$. Now consider $\bP=\bP_3^-(0,t,\infty)$ and notice that this results in a matching primal objective $r^p_{2b}(\bP,p,t)= t(w^-_2+w^-_3) - pw^-_3 = \mu - \frac{d}{2} = r^d_{2b}(\lambda,p,t)$. Since $r_{2b}(p,t)$ is constant in $t$, we get $$r_{2b}^*(p) = \mu-\frac{d}{2}.$$

We now turn to three scenarios that, due to equivalence to earlier cases, can be addressed succinctly. For $t \in (p,\mu]$, denoted as Scenario 2c, the structure is equivalent to Scenario 1c with the added restriction that $t > p$. However, since $r_{1c}(p,t)$ is increasing in $t$, this restriction does not affect the optimal value, and we obtain $r^*_{2c}(p) = r^*_{1c}(p) = (\mu-p)\frac{2\mu-d}{2\mu}$. For $t \in (\mu,\tau_2]$, denoted by Scenario 2d, the structure is equivalent to Scenario 1d, and hence, $r^*_{2d}(p) = r^*_{1d}(p) = \mu-p+\frac{dp}{2\mu}$. For $t \in (\tau_2,\infty)$, denoted by Scenario 2e, the structure is equivalent to Scenario 1e. Since $p\geq\mu-d/2$, the distribution $\bP^y_2(t)$ is always feasible. Hence, $r^*_{2e}(p) = r^*_{1e}(p) = \mu - p + \frac{dp}{2\mu}$.

Finally, maximizing over all scenarios gives 
\begin{align*}
\sup_t\sup_{\bP \in \cP(\mu,d)} t\bP(X\geq t) - p\bP(X> p) &= \max\{r^*_{2a}(p),r^*_{2b}(p),r^*_{2c}(p),r^*_{2d}(p),r^*_{2e}(p)\}\\
&=\max\{r^*_{2b}(p),r^*_{2d}(p)\}\\
&=\max\left\{\mu-\frac{d}{2},\mu-p+\frac{dp}{2\mu}\right\},
\end{align*}
since $r^*_{2a}(p) = r^*_{2b}(p)$, $r^*_{2c}(p) = (\mu-p)\frac{2\mu-d}{2\mu} \leq \mu-p+\frac{dp}{2\mu} = r^*_{2d}(p)$, and $r^*_{1e}(p) = r^*_{1d}(p)$.

\medskip
\textbf{Case 3: $p \in [\mu,\tau_2)$.} The first two scenarios are identical to previous scenarios. For $t\in(0,\tau_1]$, which we call Scenario 3a, is equivalent to Scenario 2a. Hence, $r^*_{3a}(p) = r^*_{2a}(p) = \mu - d/2$. Similarly, for $t \in (\tau_1,\mu]$, which we call Scenario 3b, is equivalent to Scenario 2b, albeit over the larger interval $(\tau_1,\mu]$. However, since $r^d_{2b}(\lambda,p,t)$ is constant in $t$, this larger interval does not matter and we obtain $r^*_{3b}(p) = r^*_{2b}(p) = \mu- d/2$.

Next, consider $t \in (\mu,p]$, which we call Scenario 3c. Consider $F(0) = \lambda_0 + \lambda_2 \mu = 0$, $F(p)=\lambda_0 + \lambda_1 p + \lambda_2(p-\mu)=t$, and $F'_+(\mu)=\lambda_1 +\lambda_2=0$. This results in
\begin{align*}
    \lambda_0 = \frac{t}{2}, \ \ \lambda_1 = \frac{t}{2\mu}, \ \ \lambda_2 = -\frac{t}{2\mu}
\end{align*}
with corresponding dual value $r^d_{3c}(\lambda,p,t) = t(1-\frac{d}{2\mu})$. Now consider $\bP=\bP_3^+(0,p,\infty)$ and notice that this results in a matching primal objective $r^p_{3c}(\bP,p,t)=t (w^+_2+w^+_3)-p w^+_3 = r^d_{3c}(\lambda,p,t)$. Clearly, $r_{3c}(p,t)$ is increasing in $t$ and hence, we only need to consider $t = p$ with $$r_{3c}^*(p) = \left(\mu-\frac{d}{2}\right)\frac{p}{\mu}.$$

We now turn to two scenarios that, due to equivalence to earlier cases, can be addressed succinctly. For $t \in (p,\tau_2]$, denoted by Scenario 3d, the structure is equivalent to Scenario 1d with the added restriction that $t > p$. However, since $r_{1d}(p,t)$ is increasing in $t$, this restriction does not affect the optimal value, and we obtain $r^*_{3d}(p) = r^*_{1d}(p) = \mu - p + \frac{dp}{2\mu}$. For $t \in (\tau_2,\infty)$, denoted by Scenario 3e, the structure is equivalent to Scenario 2e. Hence, $r^*_{3e}(p) = r^*_{2e}(p) = \mu - p + \frac{dp}{2\mu}$.

Finally, maximizing over all scenarios gives 
\begin{align*}
\sup_t\sup_{\bP \in \cP(\mu,d)} t\bP(X\geq t) - p\bP(X> p) &= \max\{r^*_{3a}(p),r^*_{3b}(p),r^*_{3c}(p),r^*_{3d}(p),r^*_{3e}(p)\}\\
&=\max\{r^*_{3c}(p),r^*_{3d}(p)\}\\
&=\max\left\{\left(\mu-\frac{d}{2}\right)\frac{p}{\mu},\mu-p+\frac{dp}{2\mu}\right\},
\end{align*}
as $r^*_{3a}(p) = r^*_{3b}(p) \leq r^*_{3c}(p)$ and $r^*_{3d}(p) = r^*_{3e}(p)$.

\medskip
\textbf{Case 4: $p \in [\tau_2,\infty)$.} Consider $\bP = \bP^x_2(0)$. Consequently, $p\bP(X> p) = 0$. Simultaneously, by selecting $t = \tau_2$ we obtain
$\tau_2\bP(X\geq \tau_2)=\left(\mu + \frac{d\mu}{2\mu-d}\right)\left(1-\frac{d}{2\mu}\right) = \mu$. Hence, since $\mu$ is a trivial upper bound of absolute regret, we conclude that $$\sup_t\sup_{\bP\in\cP(\mu,d)}t\bP(X\geq t) - p\bP(X>p) = \mu.$$ Finally, for $p \in (0,\infty)$, we remark that $\sup_t\sup_{\bP \in \cP(\mu,d)} t\bP(X\geq t) - p\bP(X> p)$ is continuous in $p$. This completes the proof.
\end{proof}

\subsection{Outer minimization problem}
Consider Lemma \ref{lem:sup_ar}. Clearly, $p \leq \tau_2$, as otherwise the regret is maximal. For notational convenience, we introduce two functions corresponding to the left-hand and right-hand branches of the maximal regret. Specifically, define
    \begin{align*}
        f_l(p) = \begin{cases}
            \frac{dp}{2(\mu-p)}, \quad& p \in (0,\tau_1],\\
           \mu-\frac{d}{2}, \quad& p \in [\tau_1,\mu],\\
            \left(\mu-\frac{d}{2}\right)\frac{p}{\mu}, \quad& p \in [\mu,\tau_2],
        \end{cases}
    \end{align*}
    and $f_r(p) = \mu - p + \frac{dp}{2\mu}$. Notice that $f_l(p)$ is increasing and continuous in $p$, satisfying $f_l(0) = 0$ and $f_l(\tau_2)=\mu$. Meanwhile, $f_r(p)$ is decreasing and continuous in $p$, satisfying $f_r(0) = \mu$ and $f_r(\tau_2) = 0$. Consequently, 
    $p^*_{ar}$ is determined by the condition $f_l(p^*_{ar}) = f_r(p^*_{ar})$. A straightforward analysis then yields \eqref{eq:regret_price} with the caveat that when $d \in [(3-\sqrt{5})\mu,\mu]$, every $p \in [\mu(\frac{d}{2\mu-d}),\mu]$ is optimal.

\subsection{Single-item absolute regret with MAD upper bound}\label{ap:regret_m1_ub}
When we replace $d$ by an upper bound of the MAD $\bar{d}$, it slightly changes Proposition \ref{prop:regret_m1}, as the optimal price is not decreasing in $d$. Therefore, we will now derive this result with $d$ replaced by an upper bound $\bar{d}$ and show that this doesn't impact Proposition \ref{th:3p_m1}. Formally, we now consider
\begin{align}
\mathcal{P}(\mu,\bar{d}) = \{\mathbb{P}: \mathbb{E}_{\mathbb{P}}[X] = \mu, \, \mathbb{E}_{\mathbb{P}}[|X - \mu|] \leq \bar{d}, \text{ and } X \in [0,\infty)\}.
\label{eq:upperbound_mad}
\end{align}

We start with solving the inner maximization problem. We will prove this by changing the proof of \ref{lem:sup_ar} at certain places.
\begin{lemma}\label{lem:sup_ar2}
    Let $\mu > 0$, $0 \leq \bar{d} < 2\mu$, and $p \in (0,\infty)$. Then
    \begin{align}
    \sup_{\bP\in\cP(\mu,\bar{d})}\textup{OPT}(\bP)-\textup{REV}(p,\bP) = \begin{cases}
        \max\left\{\frac{\bar{d}p}{2(\mu-p)},\mu-p+\frac{\bar{d}p}{2\mu}\right\}, \quad& p \in [0,\mu-\bar{d}/2], \\
        \max\left\{p,\mu-p+\frac{\bar{d}p}{2\mu}\right\}, \quad& p \in [\mu-\bar{d}/2,\mu], \\
        \mu, \quad& p \in [\mu,\infty).
    \end{cases}
\end{align}
\end{lemma}
\begin{proof}
The proof of Proposition \ref{prop:regret_m1} has four case distinctions: Case 1 with $p \in (0,\mu-d/2)$, case 2 with $p \in [\mu-d/2,\mu)$, case 3 with $p \in [\mu,\mu+d\mu/(2\mu-d))$, and case 4 with $p \in [\mu+d\mu/(2\mu-d),\infty)$.
    
Clearly, case 4 remains unchanged, as $\mu$ is still an upper bound of the absolute regret and we are maximizing over a larger set. However, Nature can now accomplish this for case 3 as well by selecting the degenerate distribution $\bP_{\mu}$ with all mass on $\mu$ so that $p\bP_{\mu}(X>p)=0$ while $\textup{OPT}(\bP_{\mu})=\sup_{t}t\bP_{\mu}(X\geq t)=\mu\bP_{\mu}(X\geq\mu)=\mu$. Furthermore, case 1 also remains unchanged, as the optimal price in Proposition \ref{prop:regret_m1} is still decreasing in $d$ when $p^*_{ar} < \mu-d/2$. Hence, one can replace $d$ by $\bar{d}$ without loss of generality in case 1.

We now turn to case 2, which requires a more careful analysis. The following class of two-point distributions is used for this:
\begin{align}\label{eq:ubd}
\bar{\bP}_2(x,y) = \left\{ \begin{array}{ll}
x & \text{ w.p. } v_1 = \frac{y-\mu}{y-x} \\
y & \text{ w.p. } v_2 = \frac{\mu-x}{y-x} 
\end{array}\right.
\end{align}
with $y \in [\mu,\frac{2\mu(\mu-x)-\bar{d}x}{2(\mu-x)-\bar{d}}]$ when $x \in [0,\mu-\bar{d}/2)$ and $y \in [\mu,\infty)$ when $x \in [\mu-\bar{d}/2,\mu]$. Notice that \eqref{eq:ubd} differs from \eqref{def:2pd} due to the MAD being an upper bound. 

We start by formulating the new primal problem:
\begin{equation}
\begin{aligned}
&\! \sup_{\mathbb{P}\in \mathcal{M}} &  &\int_x t\1_{\{x\geq t\}}-p\1_{\{x> p\}}{\rm d} \mathbb{P}(x),\\
&\text{s.t.} &      &  \int_x {\rm d}\mathbb{P}(x)=1, \ \int_x x{\rm d}\mathbb{P}(x)=\mu,\ \int_x |x-\mu|{\rm d}\mathbb{P}(x)\leq\bar{d}.
\end{aligned}
\end{equation}
This changes the corresponding dual into
\begin{equation}
\begin{aligned}
&\inf_{\lambda_0,\lambda_1 \in \R, \lambda_2\in \R_+}\quad
 \lambda_0 + \lambda_1 \mu + \lambda_2 \bar{d},\\
&\text{s.t.}\quad
 G(x) := t\1_{\{x\geq t\}}-p\1_{\{x> p\}}
 \leq \lambda_0 + \lambda_1 x + \lambda_2 |x-\mu| =: F(x),
\quad \forall x\in[0,\infty).
\end{aligned}
\end{equation}
Hence, scenarios where $\lambda_2 < 0$ are no longer dual feasible. First, consider $t \in (0,p]$, which we call Scenario 2a. Consider dual solution $\lambda_0=t$, $\lambda_1=0$, and $\lambda_2=0$. Furthermore, consider primal solution $\bP=\bar{\bP}_2(p,y)$ and let $y\xrightarrow{} \infty$. This results in $r^d_{2a}(\lambda,p,t) = t=\lim_{y\xrightarrow{}\infty}r^p_{2a}(\bP,p,t)$. Clearly, $r_{2a}(p,t)=t$ is increasing in $t$, so we only need to consider $r^*_{2a}(p)=p$.

Second, consider $t \in (p,\mu]$, which we call Scenario 2b. Consider dual solution $\lambda_0=t-p$, $\lambda_1=0$, and $\lambda_2=0$. Furthermore, consider primal solution $\bP=\bP_{\mu}$, i.e., the degenerate distribution with all probability mass on $\mu$. This results in $r^d_{2b}(\lambda,p,t)=t-p=r^p_{2b}(\bP,p,t)$. Clearly, $r_{2b}(p,t)=t-p$ is increasing in $t$, so we only need to consider $r^*_{2b}(p)=\mu-p$.

Third, consider $t \in (\mu,\mu+(\bar{d}\mu)/(2\mu-\bar{d})]$, which we call Scenario 2c. Consider dual solution $\lambda_0=0$, $\lambda_1=\frac{t-p}{2(t-\mu)}$, and $\lambda_2=0$. Furthermore, consider primal solution $\bar{\bP}_2(0,t)$. This results in $r^d_{2c}(\lambda,p,t)=(t-p)\frac{\mu}{t}=r^p_{2c}(\bP,p,t)$. Clearly, $r_{2c}(p,t)=(t-p)\frac{\mu}{t}$ is increasing in $t$, so we only need to consider $r^*_{2c}(p)=\mu-p+\frac{\bar{d}p}{2\mu}$.

Fourth, consider $t \in (\mu+(\bar{d}\mu)/(2\mu-\bar{d}),\infty)$, which we call Scenario 2d. This is equivalent to Scenario 2e in the original proof, as $\lambda_2 = \frac{t-p}{2(t-\mu)}\geq 0$. Hence, $r^*_{2d}(p)=\mu-p+\frac{\bar{d}p}{2\mu}$.

Finally, maximizing over all scenarios gives 
\begin{align*}
\sup_t\sup_{\bP \in \cP(\mu,\bar{d})} t\bP(X\geq t) - p\bP(X> p) &= \max\{r^*_{2a}(p),r^*_{2b}(p),r^*_{2c}(p),r^*_{2d}(p)\}\\
&=\max\{r^*_{2a}(p),r^*_{2c}(p)\},
\end{align*}
since $r^*_{2b}(p) = \mu-p \leq \mu-p+\frac{\bar{d}p}{2\mu} = r^*_{2c}(p)$, and $r^*_{2d}(p) = r^*_{2c}(p)$.
\end{proof}

We will now solve the outer minimization problem. 
\begin{proposition}[Single-item absolute regret with MAD upper bound]
Consider $0\leq \bar{d} < 2\mu$. The optimal absolute regret of selling one good for a deterministic price is achieved by price
    \begin{align}\label{eq:regret_price2}
    \bar{p}^*_{ar}:=\arg\inf_p\sup_{\bP\in\cP(\mu,\bar{d})}\textup{OPT}(\bP)-\textup{REV}(p,\bP)=\begin{cases}
         \mu\left(\frac{2\mu-\sqrt{2\mu \bar{d}}}{2\mu-\bar{d}}\right), \quad& \bar{d} \in [0,(3-\sqrt{5})\mu] \\
        \mu\left(\frac{2\mu}{4\mu-\bar{d}}\right), \quad& \bar{d} \in [(3-\sqrt{5})\mu,2\mu]
    \end{cases}.
\end{align}
\end{proposition}
\begin{proof}
    Consider Lemma \ref{lem:sup_ar2}. Clearly, $p \leq \mu$, as otherwise the regret is maximal. For notational convenience, we introduce two functions corresponding to the left-hand and right-hand branches of the maximal regret. Specifically, define
    \begin{align*}
        f_l(p) = \begin{cases}
            \frac{\bar{d}p}{2(\mu-p)}, \quad& p \in (0,\mu-\bar{d}/2],\\
            p, \quad& p \in [\mu-\bar{d}/2,\mu],
        \end{cases}
    \end{align*}
    and $f_r(p) = \mu - p + \frac{\bar{d}p}{2\mu}$ with $p \in (0,\mu]$. Notice that $f_l(p)$ is increasing and continuous in $p$, satisfying $f_l(0^+) = 0$ and $f_l(\mu)=\mu$. Meanwhile, $f_r(p)$ is decreasing and continuous in $p$, satisfying $f_r(0^+) = \mu$ and $f_r(\mu) = \bar{d}/2$. Consequently, $f_l(0^+)=0<\mu=f_r(0^+)$ and $f_l(\mu)=\mu>\bar{d}/2=f_r(\mu)$. Hence, $\bar{p}^*_{ar}$ is determined by the unique intersection $f_l(\bar{p}^*_{ar}) = f_r(\bar{p}^*_{ar})$. Solving for $\bar{p}^*_{ar}$ then yields \eqref{eq:regret_price2}.
\end{proof}

Clearly, $\bar{p}^*_{ar} > p^*_{cr}$, so Proposition \ref{th:3p_m1} remains valid when the MAD is replaced by an upper bound.
\end{document}